\documentclass[a4paper,12pt]{article}
\usepackage[margin=3cm]{geometry}
\usepackage{amsmath}
\usepackage{amsthm}
\usepackage{amssymb}
\usepackage{amsfonts}
\usepackage{tipa}
\usepackage{tikz}
\usetikzlibrary{arrows}
\usetikzlibrary{calc}
\usepackage{etoolbox}

\newtheorem{theorem}{Theorem}[section]

\newcommand*\sxto[1]{\xrightarrow{\smash{#1}}}
\newcommand{\ket}[1]{\ensuremath{| #1 \rangle}}
\newcommand{\bra}[1]{\ensuremath{\langle #1 |}}

\newcommand\xmapsto[1]{\mathrel{\ensuremath{\raisebox{1.5pt}{\ensuremath{{\scriptscriptstyle |}}}\ensuremath{\hspace{-5.5pt}\sxto{#1}}}}}

\newcommand\ignore[1]{}
\newcommand\question[1]{

\vspace{5pt}
\noindent
\emph{#1}

\vspace{5pt}}

\pgfdeclarelayer{background}
\pgfdeclarelayer{foreground}
\pgfsetlayers{background,main,foreground}
\tikzset{smallbox/.style={draw, fill=white, minimum height=0.45cm, minimum width=0.45cm, inner sep=-100pt}}
\tikzset{smallcircle/.style={draw, circle, fill=black, inner sep=1.5pt}}
\tikzset{circlelabel/.style={draw, circle, inner sep=-5pt, font=\scriptsize, fill=white, minimum width=12pt}}
\tikzset{smallfont/.style={font=\scriptsize}}
\tikzset{sepstyle/.style={thin, dashed}}

\makeatletter
\def\calign@preamble{%
   &\hfil\strut@
    \setboxz@h{\@lign$\m@th\displaystyle{##}$}%
    \ifmeasuring@\savefieldlength@\fi
    \set@field
    \hfil
    \tabskip\alignsep@
}
\let\cmeasure@\measure@
\patchcmd\cmeasure@{\divide\@tempcntb\tw@}{}{}{}
\patchcmd\cmeasure@{\divide\@tempcntb\tw@}{}{}{}
\patchcmd\cmeasure@{\ifodd\maxfields@
  \global\advance\maxfields@\@ne
  \fi}{}{}{}    
\newenvironment{calign}
{%
  \let\align@preamble\calign@preamble
  \let\measure@\cmeasure@
  \align
}
{%
  \endalign
}  
\makeatother

\tikzset{keymwidth/.initial=5cm}
\tikzset{mwidth/.style={keymwidth={#1}}}
\tikzset{keymheight/.initial=5cm}
\tikzset{mheight/.style={keymheight={#1}}}
\makeatletter
\pgfdeclareshape{Measurement}
{
    \savedanchor\centerpoint
    {
        \pgf@x=0pt
        \pgf@y=0pt
    }
    \savedmacro{\savedwidth}
    {
        \pgfkeysgetvalue{/tikz/keymwidth}{\pgf@rectc}
        \let\savedwidth\pgf@rectc
    }
    \savedmacro{\savedheight}
    {
        \pgfkeysgetvalue{/tikz/keymheight}{\pgf@rectc}
        \let\savedheight\pgf@rectc
    }
    \anchor{center}{\centerpoint}
    \anchorborder{\centerpoint}
    \backgroundpath
    {
        \pgfkeysgetvalue{/tikz/keymwidth}{\mwidth}
        \pgfkeysgetvalue{/tikz/keymheight}{\mheight}
        \pgfsetfillcolor{white}
        \pgfpathmoveto{\pgfpoint{-0.5*\mwidth}{-0.5*\mheight}}
        \pgfpathlineto{\pgfpoint{0.5*\mwidth}{-0.5*\mheight}}
        \pgfpathlineto{\pgfpoint{0.5*\mwidth}{0.5*\mheight}}
        \pgfpathlineto{\pgfpoint{-0.5*\mwidth}{0.5*\mheight}}
        \pgfpathlineto{\pgfpoint{-0.5*\mwidth}{-0.5*\mheight}}
        \pgfusepath{fill}
        \pgfsetstrokecolor{black}
        \pgfpathmoveto{\pgfpoint{-0.5*\mwidth}{-0.5*\mheight}}
        \pgfpathlineto{\pgfpoint{0.5*\mwidth}{-0.5*\mheight}}
        \pgfpathlineto{\pgfpoint{0.5*\mwidth}{0.5*\mheight}}
        \pgfpathlineto{\pgfpoint{-0.5*\mwidth}{0.5*\mheight}}
        \pgfpathlineto{\pgfpoint{-0.5*\mwidth}{-0.5*\mheight}}
        \pgfpathlineto{\pgfpoint{0.5*\mwidth}{-0.5*\mheight}}
        \pgfusepath{stroke}
        \pgftransformshift{\pgfpoint{0cm}{-0.2*\mheight}}
        \pgfpathmoveto{\pgfpointpolar{45}{0.0*\mheight}}
        \pgfpathlineto{\pgfpointpolar{45}{0.6*\mheight}}
        \pgfusepath{stroke}
        \pgfpathmoveto{\pgfpoint{-0.4*\mheight}{0.0*\mheight}}
        \pgfpatharc{180}{0}{0.4*\mheight}
        \pgfusepath{stroke}
    }
}
\makeatother

\newcommand\Z{\mathbb{Z}}
\newcommand\C{\ensuremath{\mathbb{C}}}

\def\swangle{-145}
\def\seangle{-35}
\def\nwangle{145}
\def\neangle{35}
\def\id{\mathrm{id}}
\def\GL{\textit{GL}}
\def\Mat{\mathrm{Mat}}
\def\Tr{\mathrm{Tr}}
\def\dim{\mathrm{d}}

\def\licsscale{0.7}

\begin{document}

\title{Topological Structure of Quantum Algorithms}

\title{The Topology of Quantum Algorithms}
\author{Jamie Vicary
\\
\texttt{jamie.vicary@cs.ox.ac.uk}
\\[10pt]
\begin{tabular}{c}
Centre for Quantum Technologies, University of Singapore
\\
and Department of Computer Science, University of Oxford
\\[0pt]
\end{tabular}}
\date{October 10, 2013}
\maketitle

\begin{abstract}
\boldmath
We use a categorical topological semantics to examine the Deutsch-Jozsa, hidden subgroup and single-shot Grover algorithms. This reveals important structures hidden by conventional algebraic presentations, and allows novel  proofs of correctness via local topological operations, giving for the first time a satisfying high-level explanation for why these procedures work. We also investigate generalizations of these algorithms,  providing improved analyses of those already in the literature, and a new generalization of the single-shot Grover algorithm.
\end{abstract}

\section{Overview}

\subsection{Introduction}

\noindent
Important quantum procedures often seem mysterious because of the low-level way in which they are presented. A direct description of the required state preparations, unitary operators and projective measurements in terms of matrices of complex numbers gives exactly the required information to actually implement a particular protocol --- but almost no information about why it should work.

One reason for this is that quantum information has a \textit{topological} nature. The overall effect of a composite of quantum operations depends not so much on the order of composition, as on the topological flows of information that this induces, as demonstrated in striking fashion by Abramsky, Coecke et al~\cite{ac04-csqp, ac08-cqm, c03-loe, cd11-iqo, cp06-pnwt, cpp09-cqs, cp06-qmws, cpv08-dfb, cp08-ecc, p10-gaqc, v12-hsqp}, in a body of work that comprises the  Categorical Quantum Mechanics research programme. Their topological semantics de-emphasize the matrices of complex numbers used in conventional presentations of quantum procedures, replacing them with geometrical primitives that give a satisfying explanation for why these procedures work. This provides  a mature set of tools for analyzing a wide range of quantum procedures, especially those which make use of Bell-type entanglement or complementary observables.

However, this research programme has not yet provided a topological analysis of the most important quantum \textit{algorithms} (although see~\cite{p10-gaqc} for interesting work along these lines.) Quantum algorithms can be usefully distinguished from quantum \emph{procedures} such as quantum teleportation, dense coding or key exchange, which bring about a certain physical effect rather than compute a particular quantity. Features of quantum  algorithms that make them hard to model using the topological formalism include the use of an oracle and a reliance on group representation theory.

This paper overcomes these difficulties, providing a topological account of the Deutsch-Jozsa, hidden subgroup and single-shot Grover algorithms. This gives  a consistent, high-level structural account of why these important quantum algorithms work, directly exposing the relevant flows of quantum information, and allowing new proofs of correctness.

This new perspective makes it much easier to consider generalizations of these algorithms. For the single-shot Grover algorithm, this generalization appears to be new. While the traditional Grover algorithm is based on the group $\mathbb{Z}_2$, our generalization is based on an arbitrary finite group $G$. Given a set coloured by the elements of $G$ in a particular ratio, the algorithm identifies with a single query an element with an infrequent colour. To give a concrete example, consider a basket of red, blue and green balls, promised to be coloured in the ratio 4:1:1 in some order. For $G=\Z_3$ the generalized Grover algorithm will identify one of the less frequently--occurring colours with a single query, a task that cannot be achieved with the ordinary Grover algorithm.

The literature already contains descriptions of generalized  Deutsch-Jozsa~\cite{h99-coqa, bdb06-etp} and hidden subgroup~\cite{hrt00-nsr} algorithms. As with the original algorithms, these generalizations are technically opaque, and require long proofs which provide relatively little insight. We obtain new descriptions of these generalized algorithms using our topological techniques, giving a clear view of their structure, and simpler proofs of correctness. For the Deutsch-Jozsa algorithm, this clarity allows us to produce a new generalization that goes beyond that given in the literature.

The topological forms of the Deutsch-Jozsa, hidden subgroup and single-shot Grover algorithms that we develop are summarized by the following diagrams:
{\def\licsscale{0.5}
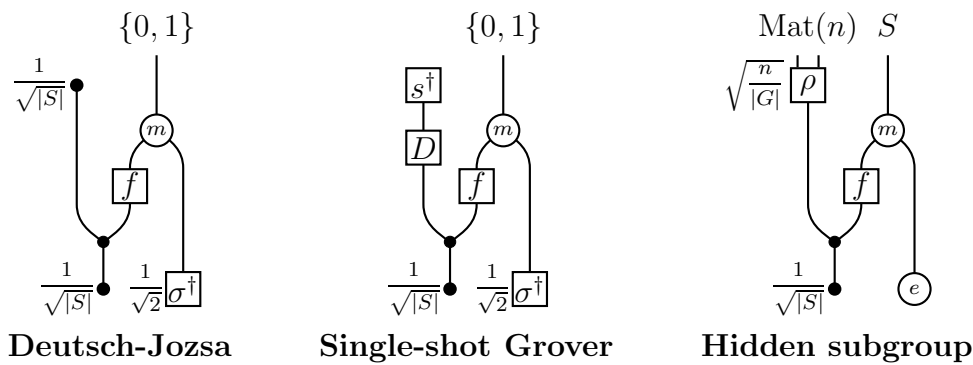
\begin{calign}
\nonumber
\begin{aligned}
\begin{tikzpicture}[thick, scale=\licsscale]
\begin{pgfonlayer}{foreground}
    \node (dot) [smallcircle] at (0,1) {};
    \node (f) [smallbox, anchor=south, thick] at (0.7,2) {$f$};
    \node (m) [circlelabel, thick] at ([xshift=0.7cm, yshift=1cm] f.north) {$m$};
\end{pgfonlayer}
\draw (0,-0.25)
        node [smallcircle] (bdot) {}
    to (0,1)
    to [out=\nwangle, in=south] (-0.7,2)
    to ([yshift=1.2cm] m.center -| -0.7,1)
        node (top) [smallcircle] {};
\draw (0,1)
    to [out=\neangle, in=south] (f.south)
    to (f.north)
    to [out=up, in=\swangle] +(0.7,1)
    to [out=\seangle, in=up] +(0.7,-1)
    to (2.1,-0.25)
        node (sigmadag) [smallbox] {$\sigma ^\dag$};
\draw (m.center) to (1.4,5.95)
        node [above] {$\{0,1\}$};
\node [anchor=east, smallfont, inner sep=0pt] at (bdot.west) {$\displaystyle \frac{1}{\sqrt{|S|}}$};
\node [anchor=east, smallfont, inner sep=0pt] at (top.west) {\makebox[0pt][r]{$\displaystyle \frac{1}{\sqrt{|S|}}$}};
\node [anchor=east, smallfont, inner sep=0pt] at (sigmadag.west) {$\displaystyle \frac{1}{\sqrt{2}}$};
\end{tikzpicture}
\end{aligned}
&
\begin{aligned}
\begin{tikzpicture}[thick, scale=\licsscale]
\begin{pgfonlayer}{foreground}
    \node (dot) [smallcircle] at (0,1) {};
    \node (f) [smallbox, anchor=south, thick] at (0.7,2) {$f$};
    \node (m) [circlelabel, thick] at ([xshift=0.7cm, yshift=1cm] f.north) {$m$};
\end{pgfonlayer}
\draw (0,-0.25)
        node [smallcircle] (bdot) {}
    to (0,1)
    to [out=\nwangle, in=south] (-0.7,2)
    to ([yshift=1.2cm] m.center -| -0.7,1)
        node (rho) [smallbox] {$s ^\dag$};
\draw (0,1)
    to [out=\neangle, in=south] (f.south)
    to (f.north)
    to [out=up, in=\swangle] +(0.7,1)
    to [out=\seangle, in=up] +(0.7,-1)
    to (2.1,-0.25)
        node [smallbox] (sigmadag) {$\sigma ^\dag$};
\draw (m.center) to (1.4,5.95)
        node [above] {$\{0,1\}$};
\node [anchor=east, smallfont, inner sep=0pt] at (bdot.west) {$\displaystyle \frac{1}{\sqrt{|S|}}$};
\node [anchor=east, smallfont, inner sep=0pt] at (sigmadag.west) {$\displaystyle \frac{1}{\sqrt{2}}$};
\node [smallbox] at (-0.7,3.5) {$D$};
\end{tikzpicture}
\end{aligned}
&
\begin{aligned}
\begin{tikzpicture}[thick, scale=\licsscale]
\begin{pgfonlayer}{foreground}
    \node (dot) [smallcircle] at (0,1) {};
    \node (f) [smallbox, anchor=south, thick] at (0.7,2) {$f$};
    \node (m) [circlelabel, thick] at ([xshift=0.7cm, yshift=1cm] f.north) {$m$};
\end{pgfonlayer}
\draw (0,-0.25)
        node [smallcircle] (bdot) {}
    to (0,1)
    to [out=\nwangle, in=south] (-0.7,2)
    to ([yshift=1.2cm] m.center -| -0.7,0.8)
        node (rho) [smallbox] {$\rho$};
\node [above] at (-0.7,5.95) {$\Mat(n)$};
\draw (0,1)
    to [out=\neangle, in=south] (f.south)
    to (f.north)
    to [out=up, in=\swangle] +(0.7,1)
    to [out=\seangle, in=up] +(0.7,-1)
    to (2.1,-0.25)
        node [circlelabel] {$e$};
\draw (m.center) to (1.4,5.95)
        node [above] {$S\vphantom($};
\node [anchor=east, smallfont, inner sep=0pt] at (bdot.west) {$\displaystyle \frac{1}{\sqrt{|S|}}$};
\node at (rho.west) [anchor=east, smallfont, inner sep=2pt] {\makebox[0pt][r]{$\displaystyle \sqrt{\frac{n}{|G|}}$}};
\draw (rho.60) to (rho.60 |- 0,5.95);
\draw (rho.120) to (rho.120 |- 0,5.95);
\end{tikzpicture}
\end{aligned}
\\
\nonumber
\text{\textbf{Deutsch-Jozsa}}
&
\text{\textbf{Single-shot Grover}}
&
\text{\textbf{Hidden subgroup}}
\end{calign}
}Each diagram represents an entire protocol, including the state preparation, unitary dynamics and measurement stages. We use the word `topological' to describe our approach since it is based on an existing topological notation for the linear algebra of finite groups and sets, which is summarized for easy reference in Appendices~A and B. We explore the three  algorithms in Sections~\ref{sec:newdeutschjozsa}, \ref{sec:grover} and \ref{sec:hiddensubgroup} respectively.

\subsubsection*{Acknowledgements}

This project grew out of discussions with Bob Coecke. I am grateful to John Baez, Bob Coecke, Brendan Fong and William Zeng for useful comments. Diagrams have been produced using the package \textit{TikZ}.

\subsection{Shortcomings of the circuit notation}
\label{sec:olddeutschjozsa}

We begin by outlining the usual circuit-based presentation of the Deutsch-Jozsa algorithm, as can be found in standard reference texts such as~\cite{nc-qcqi}. Presented in this manner the algorithm seems an idiosyncratic collection of individual parts, which come together to solve the Deutsch-Jozsa problem in an apparent minor miracle. In contrast, the new topological presentation we develop in Section~\ref{sec:newdeutschjozsa} shows the consistent structure that lies hidden beneath each part of the algorithm. Not all descriptions of the Deutsch-Jozsa algorithm have the shortcomings we describe here, but all of them are commonly encountered.

The conventional Deutsch-Jozsa algorithm involves a function
\begin{equation}
\{ 0,1 \} ^N \sxto f \{0,1\},
\end{equation}
where $\{ 0, 1 \}$ is the group of integers under addition modulo 2, and $\{0,1\}^N$ is the $N$\-fold cartesian product of this group for some natural number $N$. This function is promised to have one of two properties: it is either \emph{constant}, meaning that it takes the same value on every element of $S$; or \emph{balanced}, meaning  that it takes each possible value on exactly half of the elements of $S$. This gives us a first question:

\question{Why should ``constant'' or ``balanced'' be important properties?}

We then construct a unitary operator
\begin{equation}
(\C^2) ^N \otimes \C^2 \sxto{U_f} (\C^2) ^N \otimes \C^2
\end{equation}
from our function $f$, which is defined to act in the following way:
\begin{equation}
\label{eq:ufaction}
\ket {k}\otimes \ket b \xmapsto{U_f} \ket k \otimes \ket{b \oplus f(k)}
\end{equation}
Here $k$ is an element of $\{0,1\}^N$, and $b$ is an element of $\{0,1\}$. The symbol $\oplus$ represents addition modulo 2, which is the group structure on $\{0,1\}$. This poses another question:

\question{Does the group structure on $\{0,1\}^N$ play an essential role?}
\label{question:whatrolegroupstructure}

\noindent
It has not been relevant so far, in the definitions of constant or balanced function or in the definition of $U_f$.

We then build the following quantum circuit:
\begin{equation}
\label{eq:traditionaldeutschjozsacircuit}
\begin{aligned}
\begin{tikzpicture}[scale=0.6]
\draw [thick] (-2,0.75)
    node [anchor=east] {$\ket{0} ^N$}
    to (7,0.75)
    node [anchor=west] {$\bra{0} ^n$};
\draw [thick] (-2,-0.75)
    node [anchor=east] {$\ket{1}\phantom{{}^N}$}
    to (9,-0.75)
    node [anchor=west] {};
\node (Uf)
    [draw, minimum width=1.2cm, minimum height=1.5cm, fill=white, thick]
    at (2.5,0) {$U_f$};
\node (Fn1)
    [draw, minimum width=1.2cm, minimum height=0.6cm, fill=white, inner sep=-100pt, thick]
    at (0,0.75) {$F _{\{0,1\}^N}$};
\node (Fn2)
    [draw, minimum width=1.2cm, minimum height=0.6cm, fill=white, inner sep=-100pt, thick]
    at (5,0.75) {$F _{\{0,1\}^N} ^{\smash{-1}}$};
\node (F)
    [draw, minimum width=1.2cm, minimum height=0.6cm, fill=white, thick]
    at (0,-0.75) {$F _{\{0,1\}}$};
\node (meas) [Measurement, mwidth=1.2cm, mheight=0.6cm, thick] at (7.4,0.75) {};
\node [anchor=north, font=\scriptsize] at (7.4,0.4) {$\ket 0 ^N \bra 0 ^N$};
\end{tikzpicture}
\end{aligned}
\end{equation}
We read this from left to right, as is conventional for quantum circuit diagrams. We begin by preparing $N$ qubits in the state \ket 0, and a single qubit in the state \ket 1. We then apply the unitary operators $F_{\{0,1\}^N}$ and $F_{\{0,1\}}$, representing the Fourier transform operations on the groups $\{0,1\} ^N$ and $\{0,1\}$ respectively. The operator $U_f$ is then applied, followed by an inverse Fourier transform $\smash{F _{\{0,1\}^N} ^{-1}}$ on the upper family of $N$ qubits, and a projective measurement onto the state $\ket 0 ^N$. The lower qubit plays no role after the application of the unitary~$U_f$.

This seems to suggest that the group structure on $\{0,1\}^N$ is relevant after all, since it is used to construct the Fourier transform operator $F_ {\{0,1\}^N}$. Other descriptions of the Deutsch-Jozsa algorithm involve an application of the $N$-fold Hadamard gate $H ^{\otimes N}$, but this again implies the codomain of $f$ is of the form $\{0,1\}^N$. However, we only apply this operator to the element $\ket{0} ^N$, which represents the identity element of the group. On this element, the Fourier transform acts as follows, creating an even superposition of the group elements:
\begin{equation}
\label{eq:superpositionstate}
F _{\{0,1\} ^N} (\ket 0 ^N) = \frac{1}{\sqrt{2^N}} \sum _{k \in \{0,1\}^N} \ket k
\end{equation}
But this superposition is \textit{independent} of the group structure, since the Fourier transform on any abelian group gives this result on the identity element. The application of \smash{$F _{ \{0,1\}^N} ^{-1}$} at the end of the protocol is insensitive to the group structure for the same reason, since it is followed by a projective measurement onto the state $\ket 0 ^N$, and hence its sole function is to allow an effective projective measurement onto the same superposition state~\eqref{eq:superpositionstate}.  

Further questions are also raised by the quantum circuit given above.
\question{Why should we prepare the upper qubits in the state $\ket 0 ^N$?}

\noindent
A partial answer to this question is given by the reasoning above: by the action of the Fourier transform, it allows us to access an even superposition of every element of $\{0,1\}^N$, and hence, in some sense, probe our function $f$ on ``every input simultaneously''. But this is a vague notion: for what high-level reason should it enable the protocol to succeed? We can ask a similar question for the lower qubit:

\question{Why should we prepare the lower qubit in the state \ket 1?}

\noindent
The Fourier transform $F_ {\{0,1\}}$ acts on this state to produce the superposition $\frac{1}{\sqrt{2}} (\ket 0 - \ket 1)$. Why should this be the correct state to feed into the operator $U_f$?

The last step of the algorithm is to perform a projective measurement on the upper family of qubits onto the state~$\ket 0 ^N$. As discussed, given that this follows an application of the inverse Fourier transform~\smash{$F _{ \{0,1\}^N} ^{-1}$}, the overall effect is to measure the qubits in the state~\eqref{eq:superpositionstate}. If the projective measurement is successful, then we can conclude with certainty that our original function $f$ is balanced. If the measurement fails, then we can conclude with certainty that $f$ is constant.

Some further questions naturally raise themselves here.
\question{Why should we measure the upper qubits in the state $\ket 0 ^N$?}

\vspace{-5pt}

\question{Why should we expect the measurement to succeed or fail exactly when the original function was constant or balanced?}

\vspace{-5pt}
\question{Should measurement on the second family of qubits play a role?}

\noindent
A high-level understanding of the Deutsch-Jozsa algorithm will provide good answers to all these questions, as we see in the next section.

More ambitiously, we could also ask the following final question.

\question{Can this algorithm be generalized?}

\noindent
This traditional presentation gives no indication as to whether this is the case. In fact, there is a broad class of possible generalizations, and their structure emerges naturally from our new topological approach.

\section{Topological Deutsch-Jozsa}
\label{sec:newdeutschjozsa}

\subsection{Introduction}

\noindent
We now present a new topological perspective on the Deutsch-Jozsa algorithm, making use of the topological formalism for algebra introduced in Appendix~\ref{sec:topologicalalgebra}. This formalism will make clear many of the  mysterious features of the traditional presentation of the algorithm, as highlighted in the account in Section~\ref{sec:olddeutschjozsa}, by showing how the functionality of the algorithm is enabled by its topological structure.

The traditional presentation focused on properties of a function $\{0,1\}^N \sxto {\scriptstyle f} \{0,1\}$ between abelian groups. However, we argued that at no point in the protocol was the group structure on $\{0,1\}^N$ used in an essential way. For this reason, in our new perspective, we redefine our function to be of type
\begin{equation}
S \sxto {\displaystyle f} \{0,1\}
\end{equation}
where $S$ is a finite set.

The overall structure of the algorithm is given by the topological diagram below, in which time flows from bottom to top. The preparation, unitary dynamics and measurement phases are clearly indicated.
\begin{equation}
\begin{aligned}
\begin{tikzpicture}[thick, scale=\licsscale, scale=0.8]
\begin{pgfonlayer}{foreground}
    \node (dot) [smallcircle] at (0,1) {};
    \node (f) [smallbox, anchor=south, thick] at (0.7,2) {$f$};
    \node (m) [circlelabel, thick] at ([xshift=0.7cm, yshift=1cm] f.north) {$m$};
\end{pgfonlayer}
\draw (0,-0.25)
        node (bd) [smallcircle] {}
    to (0,1)
    to [out=\nwangle, in=south] (-0.7,2)
    to ([yshift=1.4cm] m.center -| -0.7,1)
        node (ts) [smallcircle] {}
        node [anchor=east, inner sep=6pt, smallfont] {$\displaystyle \frac{1}{\sqrt{|S|}}$};
\draw (0,1)
    to [out=\neangle, in=south] (f.south)
    to (f.north)
    to [out=up, in=\swangle] +(0.7,1)
    to [out=\seangle, in=up] +(0.7,-1)
    to (2.1,-0.25)
        node (sdag) [smallbox] {$\sigma ^\dag$};
\draw (m.center) to +(0,1.75)
        node [above] {$\{0,1\}$};
\node [anchor=east, smallfont] at (bd.west) {$\displaystyle \frac{1}{\sqrt{|S|}}$};
\node [anchor=east, smallfont] at (sdag.west) {$\displaystyle \frac{1}{\sqrt{2}}$};
\node at (3.5,-0.25) [anchor=west] {\textit{Preparation}};
\draw [sepstyle] (-2.3,0.5) to (7.3,0.5);
\node at (f.center -| 3.5,0) [anchor=west] {\textit{Dynamics}};
\draw [sepstyle] (-2.3,4.5) to (7.3,4.5);
\node at (ts.center -| 3.5,0) [anchor=west] {\textit{Measurement}};
\end{tikzpicture}
\end{aligned}
\end{equation}
Over the course of this section we explain why this diagram represents the conventional Deutsch-Jozsa algorithm, and show how its form gives rise to a topological proof of correctness.

We then show  how we can replace the group $\{0,1\}$ by an arbitrary finite group, and obtain a natural generalization of the Deutsch-Jozsa procedure. The algorithm we describe is close to that of Batty, Braunstein, Duncan and H\o yer~\cite{h99-coqa, bdb06-etp}, but seems strictly more general. The topological approach gives rise to a new proof of correctness, which is both shorter and more structurally informative than that allowed by the original approach.

\subsection{Constant and balanced functions}

\noindent
For the function $S \sxto {\scriptstyle f} \{0,1\}$ to be \textit{constant} means that it must factor via the 1-element set:
\begin{equation}
\label{eq:ffactor}
\begin{aligned}
\begin{tikzpicture}
\node [anchor=east] (S) at (0,0) {$S$};
\node (1) at (1.2,0) {$1$};
\node [anchor=west] (G) at (2.4,0) {$\{0,1\}$};
\draw [->] (S) to node [auto, swap] {} (1);
\draw [->] (1) to node [auto, swap] {$x$} (G);
\draw [->] (S) to [out=45, in=135] node [auto] {$f$} (G);
\end{tikzpicture}
\end{aligned}
\end{equation}
The function $S \to 1$ is the unique function to the 1\-element set, and the function \mbox{$1 \sxto {\scriptstyle x} \{0,1\}$} selects the element of $\{0,1\}$ which is the image of the function~$f$. Linearizing these functions to produce linear maps between the free vector spaces on the original sets, expression~\eqref{eq:ffactor} has the following graphical representation:
\begin{equation}
\label{eq:fconstant}
\begin{aligned}
\begin{tikzpicture}[thick, scale=\licsscale]
\draw (0,0.3) node [below] {$S$} to (0,3) node [above] {$\{0,1\}$};
\node [smallbox] at (0,1.65) {$f$};
\end{tikzpicture}
\end{aligned}
\quad=\quad
\begin{aligned}
\begin{tikzpicture}[thick, scale=\licsscale]
\draw (0,0.3) node [below] {$S$}
    to (0,1);
\node [smallcircle] at (0,1) {};
\draw (0,2) to (0,3) node [above] {$\{0,1\}$};
\node [smallbox] at (0,2) {$x$};
\end{tikzpicture}
\end{aligned}
\end{equation}
The function $f$ is constant iff there exists some function $x$ satisfying this condition.

\renewcommand\matrix[1]{\begin{pmatrix}#1\end{pmatrix}}
We now consider the case that $f$ is balanced, meaning it takes each possible value on an equal number of elements of its domain. We can express this condition with the following equation:
\begin{equation}
\matrix{1 & -1} \circ \sum _{s \in S} \ket {f(s)} = 0
\end{equation}
Here we have composed with the matrix \mbox{$\{0,1\} \sxto{\scriptstyle ( \begin{smallmatrix} 1 & -1 \end{smallmatrix} )} \C$}, which contributes $+1$ to the sum for each element $s \in S$ with $f(s)=\ket 0$, and $-1$ for each element with $f(s) = \ket 1$. So it is clear the sum is zero exactly when $f$ is balanced. We can reexpress this equation as follows:
\begin{equation}
\matrix{1 & -1} \circ f \circ \sum _{s \in S} \ket {s} = 0
\end{equation}
We can express this graphically using our topological formalism, where we write $\sigma$ for the linear map \mbox{$\{0,1\} \sxto{\scriptstyle ( \begin{smallmatrix} 1 & -1 \end{smallmatrix} )} \C$}, which is in fact the alternating representation of the group:
\begin{equation}
\label{eq:fbalanced}
\begin{aligned}
\begin{tikzpicture}[thick, scale=\licsscale]
\draw (0,-0.5) node [smallcircle] {} to (0,2.4) node [smallbox, anchor=south] {$\sigma$};
\draw (0,1) node [smallbox] {$f$};
\node at (0.0,0.1) [anchor=west] {$S$};
\node at (0.0,1.85) [anchor=west] {$\{0,1\}$};
\end{tikzpicture}
\end{aligned}
\hspace{10pt}=\quad
0
\end{equation}
Equations~\eqref{eq:fconstant} and \eqref{eq:fbalanced} can thus serve to replace the  constant and balanced properties, giving us entirely topological conditions on the function~$f$.

\subsection{Building the unitary}
\label{sec:buildingtheunitary}

At the heart of the Deutsch-Jozsa algorithm is the unitary operator $U_f$, with action defined by expression~\eqref{eq:ufaction}. We can define this topologically as follows:
\begin{equation}
\label{eq:uftopological}
\begin{aligned}
\begin{tikzpicture}[thick, scale=\licsscale]
\draw (0,0) node [below] {$S\vphantom{\{}$} to (0,4.1) node [above] {$S\vphantom{\{}$};
\draw (2,0) node [below] {$\{0,1\}$} to (2,4.1) node [above] {$\{0,1\}$};
\node [draw, minimum width=2.0cm, minimum height=2cm, fill=white] at (1,2.05) {$U_f$};
\end{tikzpicture}
\end{aligned}
\quad:=\quad
\begin{aligned}
\begin{tikzpicture}[thick, scale=\licsscale]
\begin{pgfonlayer}{foreground}
    \node (dot) [smallcircle] at (0,1) {};
    \node (f) [smallbox, anchor=south, thick] at (0.7,2) {$f$};
    \node (m) [circlelabel, thick] at ([xshift=0.7cm, yshift=1cm] f.north) {$m$};
\end{pgfonlayer}
\draw (0,0.25)
        node [below] {$S\vphantom{\{}$}
    to (0,1)
    to [out=\nwangle, in=south] (-0.7,2)
    to ([yshift=0.75cm] m.center -| -0.7,1)
        node [above] {$S\vphantom{\{}$};
\draw (0,1)
    to [out=\neangle, in=south] (f.south)
    to (f.north)
    to [out=up, in=\swangle] +(0.7,1)
    to [out=\seangle, in=up] +(0.7,-1)
    to (2.1,0.25)
        node [below] {$\{0,1\}$};
\draw (m.center) to +(0,0.75)
        node [above] {$\{0,1\}$};
\end{tikzpicture}
\end{aligned}
\end{equation}
To show this acts in the correct way, we evaluate its effect on a general input element, for some choice of $s \in S$ and $b \in \{0,1\}$:
\allowdisplaybreaks
\vspace{1cm}
\begin{align}
\nonumber
\ket s \otimes \ket b 
\quad&\xmapsto{\begin{tikzpicture}[thick, scale=\licsscale]
\draw [use as bounding box, draw=none] (-0.5,0) rectangle (1.1,1);
\draw (0,0) to (0,0.5);
\draw (0,0.5) to [out=\neangle, in=south] (0.4,1);
\draw (0,0.5) to [out=\nwangle, in=south] (-0.4,1);
\node [smallcircle] at (0,0.5) {};
\draw (1,0) to (1,1);
\end{tikzpicture}}
\quad
\ket s \otimes \ket s \otimes \ket b
\\[1cm]
\nonumber
&\xmapsto {\begin{tikzpicture}[thick, scale=\licsscale]
\draw [use as bounding box, draw=none] (-0.8,0) rectangle (0.8,1);
\draw (0,0) to (0,1);
\draw (-0.7,0) to (-0.7,1);
\draw (0.7,0) to (0.7,1);
\node [smallbox] at (0,0.5) {$f$};
\end{tikzpicture}}
\quad \ket s \otimes \ket {f(s)} \otimes \ket b
\\[1cm]
&\xmapsto{\begin{tikzpicture}[thick, xscale=-1, yscale=-1, scale=\licsscale]
\draw [use as bounding box, draw=none] (-0.5,0) rectangle (1.1,1);
\draw (0,0) to (0,0.5);
\draw (0,0.5) to [out=\neangle, in=south] (0.4,1);
\draw (0,0.5) to [out=\nwangle, in=south] (-0.4,1);
\node [circlelabel] at (0,0.5) {$m$};
\draw (1,0) to (1,1);
\end{tikzpicture}}
\quad \ket s \otimes \ket{ f(s) \oplus b}
\end{align}
This matches our earlier definition~\eqref{eq:ufaction}.

\subsection{Performing the algorithm}

To obtain a topological expression for the result of the algorithm, we precompose expression~\eqref{eq:uftopological} with the choice of input state. The input states have the following graphical form:
\begin{equation}
\label{eq:topologicalinput}
\begin{aligned}
\begin{tikzpicture}[thick, scale=\licsscale]
\draw (0,0) node (dot) [smallcircle] {} to (0,1);
\draw (2,0) node (s) [smallbox] {$\sigma ^\dag$} to +(0,1);
\node at (dot.west) [anchor=east] {$\displaystyle \frac{1}{\sqrt{|S|}}$};
\node at (s.west) [anchor=east] {$\displaystyle \frac{1}{\sqrt{2}}$};
\end{tikzpicture}
\end{aligned}
\end{equation}
This expression includes the action of the Fourier transform operations on the input states of the traditional Deutsch-Jozsa circuit~\eqref{eq:traditionaldeutschjozsacircuit}, which as we remarked in Section~\ref{sec:olddeutschjozsa} have the effect of preparing the equal superposition state and the adjoint of the irreducible representation $\sigma$, up to a normalizing factor.

By composing expressions~\eqref{eq:topologicalinput} and \eqref{eq:uftopological} we obtain the following topological expression for the quantum state after the protocol has been implemented, but before any measurement is performed:
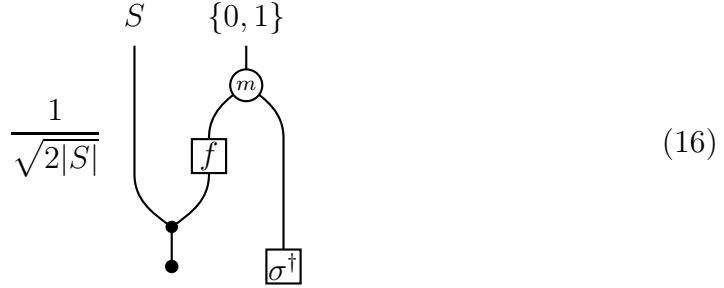
\begin{equation}
\label{eq:originalsetup}
\begin{aligned}
\begin{tikzpicture}[thick, scale=\licsscale]
\begin{pgfonlayer}{foreground}
    \node (dot) [smallcircle] at (0,1) {};
    \node (f) [smallbox, anchor=south, thick] at (0.7,2) {$f$};
    \node (m) [circlelabel, thick] at ([xshift=0.7cm, yshift=1cm] f.north) {$m$};
\end{pgfonlayer}
\draw (0,0.25)
        node [smallcircle] {}
    to (0,1)
    to [out=\nwangle, in=south] (-0.7,2)
    to ([yshift=0.75cm] m.center -| -0.7,1)
        node [above] {$S\vphantom{\{}$};
\draw (0,1)
    to [out=\neangle, in=south] (f.south)
    to (f.north)
    to [out=up, in=\swangle] +(0.7,1)
    to [out=\seangle, in=up] +(0.7,-1)
    to (2.1,0.25)
        node [smallbox] {$\sigma ^\dag$};
\draw (m.center) to +(0,0.75)
        node [above] {$\{0,1\}$};
\node [anchor=east] at (current bounding box.west) {$\displaystyle \frac{1}{\sqrt{2 |S|}}$};
\end{tikzpicture}
\end{aligned}
\end{equation}
We omit the inverse Fourier transform present in the traditional circuit diagram~\eqref{eq:traditionaldeutschjozsacircuit} as we view it as forming part of the measurement. Due to property~\eqref{eq:copyonleg} of the multiplication vertex, we can rewrite this to obtain the following equivalent expression:
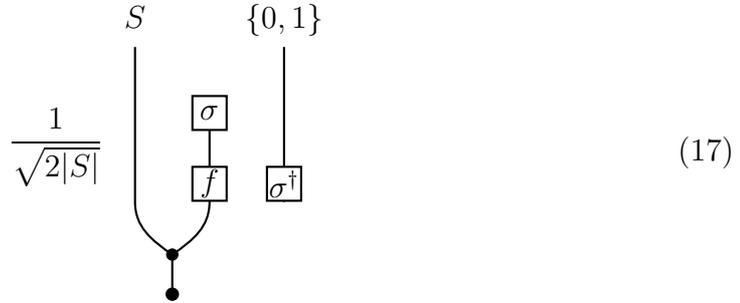
\begin{equation}
\label{eq:afteralgorithm}
\begin{aligned}
\begin{tikzpicture}[thick, scale=\licsscale]
\begin{pgfonlayer}{foreground}
    \node (dot) [smallcircle] at (0,1) {};
    \node (f) [smallbox, anchor=south, thick] at (0.7,2) {$f$};
\end{pgfonlayer}
\draw (0,0.25)
        node [smallcircle] {}
    to (0,1)
    to [out=\nwangle, in=south] (-0.7,2)
    to ([yshift=1.25cm] m.center -| -0.7,1)
        node [above] (S) {$S\vphantom{\{}$};
\draw (0,1)
    to [out=\neangle, in=south] (f.south)
    to (f.north)
    to +(0,1) node [smallbox] {$\sigma$};
\draw (2.1,2)
        node [smallbox, anchor=south] {$\sigma ^\dag$}
    to (2.1,2.35 |- S.south)
        node [above] {$\{0,1\}$};
\node [anchor=east] at (current bounding box.west) {$\displaystyle \frac{1}{\sqrt{2 |S|}}$};
\end{tikzpicture}
\end{aligned}
\end{equation}
This looks superficially different to equation~\eqref{eq:copyonleg}  only because the representation is 1\-dimensional, and so the intermediate wires become trivial.

We see that the quantum state becomes a product state. We now examine the effect of a measurement on this state in the case that the function $f$ is constant or balanced.

\subsection{Result when $f$ is constant}
\label{sec:constant}

When $f$ is constant, we apply the topological equation~\eqref{eq:fconstant} to rewrite expression~\eqref{eq:afteralgorithm} for the quantum state after the unitary evolution stage of the algorithm has been completed, obtaining the following diagram:
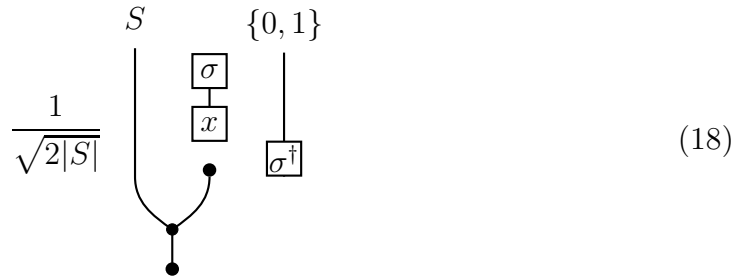
\begin{equation}
\label{eq:fconstantsimplified}
\begin{aligned}
\begin{tikzpicture}[thick, scale=\licsscale]
\begin{pgfonlayer}{foreground}
    \node (dot) [smallcircle] at (0,1) {};
    \node (f) [smallcircle, anchor=south, thick] at (0.7,2) {};
\end{pgfonlayer}
\draw (0,0.25)
        node [smallcircle] {}
    to (0,1)
    to [out=\nwangle, in=south] (-0.7,2)
    to ([yshift=0.75cm] m.center -| -0.7,1)
        node [above] (S) {$S\vphantom{\{}$};
\draw (0,1)
    to [out=\neangle, in=south] (f.south);
\draw (2.1,2)
        node [smallbox, anchor=south] {$\sigma ^\dag$}
    to +(0,2.35)
        node [above] {$\{0,1\}$};
\node [anchor=east] at (current bounding box.west) {$\displaystyle \frac{1}{\sqrt{2 |S|}}$};
\draw (0.7,3)
        node [smallbox] {$x$}
    to (0.7,4)
        node [smallbox] {$\sigma$};
\end{tikzpicture}
\end{aligned}
\end{equation}
The composite $\sigma \circ x$ represents the value of the irreducible representation $\sigma$ on the element $x \in \{0,1\}$ which is the image of the function $f$. Since $\sigma = (\,1\,\,\,\,{-1}\,)$, this composite equals~$\pm 1$. Using this fact, along with the properties of the topological algebra on a finite set, we simplify diagram~\eqref{eq:fconstantsimplified} to obtain the following expression:
\begin{equation}
\label{eq:fconstantsimplified2}
\begin{aligned}
\begin{tikzpicture}[thick, scale=\licsscale]
\begin{pgfonlayer}{foreground}
\end{pgfonlayer}
\draw (0,0)
        node [smallcircle] {}
    to (0,2)
        node [above] {$S\vphantom{\{}$};
\draw (2,0)
        node [smallbox] {$\sigma ^\dag$}
    to (2,2)
        node [above] {$\{0,1\}$};
\node [anchor=east] at (current bounding box.west) {$\displaystyle \frac{\pm 1}{\sqrt{2 |S|}}$};
\end{tikzpicture}
\end{aligned}
\end{equation}
This is a product state, with the left-hand family of qubits in the superposition state~\eqref{eq:superpositionstate}. As a result, the projective measurement performed onto the superposition state is guaranteed to be successful. This demonstrates that, in the case that $f$ is constant, we will measure the system $S$ to be in the superposition state~\eqref{eq:superpositionstate} with certainty.

\subsection{Result when $f$ is balanced}
\label{sec:balanced}

We now consider the result of a measurement of the left-hand system in the superposition state~\eqref{eq:superpositionstate} in the case that $f$ is balanced. The following diagram represents a projector onto the 1\-dimensional subspace of $S$ corresponding to the uniform superposition state:
\begin{equation}
\frac{1}{|S|}
\begin{aligned}
\begin{tikzpicture}[thick, scale=\licsscale]
\draw (0,0.3) node [smallcircle] {} to (0,1) node [above] {$S$};
\draw (0,-0.3) node [smallcircle] {} to (0,-1) node [below] {$S$};
\end{tikzpicture}
\end{aligned}
\end{equation}
Postcomposing this with expression~\eqref{eq:afteralgorithm} representing the quantum state after the main body of the algorithm has been performed, and then simplifying the result, we obtain the following diagram:
\begin{equation}
\begin{aligned}
\begin{tikzpicture}[thick, scale=\licsscale, xscale=0.8]
\begin{pgfonlayer}{foreground}
    \node (dot) [smallcircle] at (0,1) {};
    \node (f) [smallbox, anchor=south, thick] at (0.7,2) {$f$};
\end{pgfonlayer}
\draw (0,1)
    to [out=\neangle, in=south] (f.south)
    to (f.north)
    to +(0,1)
        node (sigma) [smallbox] {$\sigma$};
\draw (2.1,2)
        node [smallbox, anchor=south] {$\sigma ^\dag$}
    to (2.1,5.75)
        node [above] {$\{0,1\}$};
\draw (0,0.25)
        node [smallcircle] {}
    to (0,1)
    to [out=\nwangle, in=south] (-0.7,2)
    to (-0.7,4.5)
        node [smallcircle] {};
\draw (-0.7,5)
        node [smallcircle] {}
    to (-0.7,5.75)
        node [above] {$S\vphantom{\{}$};
\node [anchor=east, smallfont] at (current bounding box.west) {$\displaystyle \frac{1}{\sqrt{2 |S|^3}}$};
\end{tikzpicture}
\end{aligned}
\hspace{0pt}=\hspace{0pt}
\begin{aligned}
\begin{tikzpicture}[thick, scale=\licsscale, xscale=0.8]
\begin{pgfonlayer}{foreground}
    \node (dot) [smallcircle, black] at (0.7,1) {};
    \node (f) [smallbox, anchor=south, thick] at (0.7,2) {$f$};
\end{pgfonlayer}
\draw (0.7,1)
    to (f.north)
    to +(0,1)
        node (sigma) [smallbox] {$\sigma$};
\draw (2.1,2)
        node [smallbox, anchor=south] {$\sigma ^\dag$}
    to (2.1,5.75)
        node [above] {$\{0,1\}$};
\draw (-0.7,5)
        node [smallcircle] {}
    to (-0.7,5.75)
        node [above] {$S\vphantom{\{}$};
\begin{pgfscope}
\node at (0.7,0.25) [smallcircle, fill=white, draw=white] {};
\end{pgfscope}
\node [anchor=east, smallfont] at (current bounding box.west) {$\displaystyle \frac{1}{\sqrt{2 |S|^3}}$};
\end{tikzpicture}
\end{aligned}
\end{equation}
By equation~\eqref{eq:fbalanced} defining the case when $f$ is balanced, we see that this composite is equal to~0. We conclude that there is no  probability of measuring the system to be in the uniform superposition state~\eqref{eq:superpositionstate}. This completes the topological demonstration of correctness of the Deutsch-Jozsa algorithm.

\subsection{Generalization to arbitrary finite groups}
\label{sec:gendeutschjozsa}

Our topological approach allows the Deutsch-Jozsa to be easily generalized from the group $\{0,1\}$ to an arbitrary finite groups $G$, with the function then having the type
\begin{equation}
S \sxto f G .
\end{equation}
In the traditional Deutsch-Jozsa algorithm described earlier in this section, the alternating representation $\{0,1\} \sxto \sigma \C$ plays a central role. In this generalization we replace it with an arbitrary irreducible representation $G \sxto \rho \Mat(n)$, where $n$ is the dimension of the representation. We continue our convention of viewing the symbol $G$ as representing both a finite group and the free vector space $\C[G]$ depending on context. The concepts of constant and balanced functions will also be generalized, in a way which we will see below.

The generalization of the Deutsch-Jozsa algorithm presented here is closely related to that of Batty, Braunstein, Duncan and H\o yer~\cite{h99-coqa,bdb06-etp}. We believe our formulation seems strictly more general. In addition, thanks to the topological calculus, the proof of correctness is only a couple of pages, rather than the many pages required in these references. The topological formalism also makes it clear exactly \textit{why} the generalized algorithm is successful.

We begin by defining generalizations of the constant and balanced properties. For each isomorphism class of irreducible representation $\rho$, in which the group acts on the Hilbert space $V_\rho$, we pick an orthogonal projector $V_\rho \sxto {P_\rho} V_\rho$ onto some subspace, such that at least one of these projectors is nonzero. With respect to such a family $P_\rho$ of projectors, we make the following definitions of the properties \textit{$P_\rho$\-balanced} and \textit{$P_\rho$\-constant}:
\begin{calign}
\label{eq:genbalancedconstant}
\forall \rho \quad
\begin{aligned}
\begin{tikzpicture}[thick, scale=\licsscale]
\draw (0,0.25) node [smallcircle] {} node [below, white] {$S$} to (0,1.0) node [smallbox] {$f$} to (0,2) node (rho) [smallbox] {$\rho$};
\draw (rho.55) to (rho.55 |- 0,3) node [smallbox, inner sep=1pt, minimum width=0pt] {$P_\rho$} to (rho.55 |- 0,3.75);
\draw (rho.125) to (rho.125 |- 0,3.75);
\node at (0,3.75) [above] {\makebox[0pt]{$\Mat( \dim (\rho))$}};
\end{tikzpicture}
\end{aligned}
\hspace{5pt}=\hspace{5pt}
0
&
\forall \rho \,\, \exists x_\rho 
\quad
\begin{aligned}
\begin{tikzpicture}[thick, scale=\licsscale]
\draw (0,0.25) node [below] {$S$} to (0,1.0) node [smallbox] {$f$} to (0,2) node (rho) [smallbox] {$\rho$};
\draw (rho.55) to (rho.55 |- 0,3) node [smallbox, inner sep=1pt, minimum width=0pt] {$P_\rho$} to (rho.55 |- 0,3.75);
\draw (rho.125) to (rho.125 |- 0,3.75);
\node at (0,3.75) [above] {\makebox[0pt]{$\Mat( \dim (\rho))$}};
\end{tikzpicture}
\end{aligned}
\hspace{7pt}=\hspace{7pt}\begin{aligned}
\begin{tikzpicture}[thick, scale=\licsscale]
\draw (0,0.25) node [below] {$S$} to (0,1.5) node [smallcircle] {};
\node (x) at (0,2.5) [smallbox] {$x_\rho$};
\draw (x.125) to (x.125 |- 0,3.75);
\draw (x.55) to (x.55 |- 0,3.75);
\node [above] at (0,3.75) {\makebox[0pt]{$\Mat(\dim(\rho))$}};
\end{tikzpicture}
\end{aligned}
\\
\nonumber
\text{$P_\rho$-balanced}
&
\text{$P_\rho$-constant}
\end{calign}
We then define an element $\phi$ of the group algebra of $G$ in the following way:
\begin{equation}
\begin{aligned}
\begin{tikzpicture}[thick, scale=\licsscale]
\begin{pgfonlayer}{foreground}
\draw [thick] (2.1,0.5) node [above] {$G$} to (2.1,-0.75)
    node [smallbox, minimum width=17pt] {$\phi$};
\end{pgfonlayer}
\node at (2.1,-0.25) [smallbox, minimum width=17pt, white] (rhodag) {$\phi$};
\node [white] (prho) at ([yshift=-0.3cm] rhodag.-55) [anchor=north, smallbox, minimum width=13pt] {$P_\rho$};
\draw [white] ([yshift=-0.0cm] prho.south) to +(0,-0.25) to [out=down, in=down, looseness=2] ([yshift=-0.25cm] prho.south -| rhodag.-125);
\end{tikzpicture}
\end{aligned}
\quad:=\quad
\sum _\rho c_\rho \,\, 
\begin{aligned}
\begin{tikzpicture}[thick, scale=\licsscale]
\draw (2.1,0.5) node [above] {$G$} to (2.1,-0.25)
    node [smallbox, minimum width=17pt] (rhodag) {$\rho ^\dag$};
\node (prho) at ([yshift=-0.3cm] rhodag.-55) [anchor=north, smallbox, minimum width=13pt] {$P_\rho$};
\draw (rhodag.-55) to (prho.north);
\draw (prho.south) to +(0,-0.25) to [out=down, in=down, looseness=2] ([yshift=-0.25cm] prho.south -| rhodag.-125) to (rhodag.-125);
\end{tikzpicture}
\end{aligned}
\end{equation}
The summation is over the equivalence classes of irreducible representations $\rho$ of $G$. The coefficients $c_\rho$ are arbitrary but nonzero, and can always be picked to ensure that $\phi$ has norm 1. It is here that our procedure is more general than that of Batty, Duncan and Braunstein~\cite{bdb06-etp}, who require a single $P_\rho$ to have trace 1, and all others to have rank 0. Also, they require the nonzero projector to have support given by a particular basis element of the representation space.

Our generalized Deutsch-Jozsa algorithm proceeds as in the standard case described earlier in this section, except we use this state $\phi$ as the initial state for the system $G$. After the unitary dynamics step, the state of the system has the following topological form:
\def\tempxscale{1}
\begin{equation}
\label{eq:gendjafterunitary}
\sum_\rho c_\rho
\begin{aligned}
\begin{tikzpicture}[thick, xscale=\tempxscale, scale=\licsscale, scale=0.8]
\begin{pgfonlayer}{foreground}
    \node (dot) [smallcircle] at (0,1) {};
    \node (f) [smallbox, anchor=south, thick] at (0.7,2) {$f$};
    \node (m) [circlelabel, thick] at ([xshift=0.7cm, yshift=1.9cm] f.north) {$m$};
\end{pgfonlayer}
\draw (0,-0.0)
        node [smallcircle] (dot2) {}
    to (0,1)
    to [out=\nwangle, in=south] (-0.7,2)
    to ([yshift=1.0cm] m.center -| -0.7,0.75)
        node [above] {$S$};
\draw (0,1)
    to [out=\neangle, in=south] (f.south)
    to (f.north) to +(0,0.9)
    to [out=up, in=\swangle] (m.center)
    to [out=\seangle, in=up] +(0.7,-1)
    to (2.1,1)
        node [smallbox, minimum width=17pt] (rhodag) {$\rho ^\dag$};
\node (prho) at ([yshift=-0.3cm] rhodag.-55) [anchor=north, smallbox, minimum width=13pt] {$P_\rho$};
\draw (m.center) to +(0,1.0)
        node [above] {$G$};
\node [anchor=east, smallfont] at (dot2.west) {\makebox[0pt][r]{$\displaystyle \frac{1}{\sqrt{|S|}}$}};
\draw (rhodag.-55) to (prho.north);
\draw (prho.south) to [out=down, in=down, looseness=2] (prho.south -| rhodag.-125) to (rhodag.-125);
\end{tikzpicture}
\end{aligned}
\hspace{-5pt}
\quad=\quad
\sum_\rho c_\rho
\begin{aligned}
\begin{tikzpicture}[thick, xscale=-1, xscale=\tempxscale, scale=\licsscale, scale=0.8]
\draw (1,-1.5) node (f) [smallbox] {$f$} to (1,-0.4) node (r2) [smallbox] {$\rho$};
\draw (0.3, 0.9) node (r3) [smallbox] {$\rho ^\dag$};
\draw (r2.120) to [out=up, in=down, in looseness=0.7] (r3.-120);
\draw (r3.north) to (0.3,1.7) node [above] {$G$};
\draw (f.south) to [out=down, in=\nwangle] +(0.7,-1) node [smallcircle] (dot) {} to +(0,-1) node [smallcircle] {} node [anchor=east, smallfont] {\makebox[0pt][r]{$\displaystyle \frac{1}{\sqrt{|S|}}$}};
\draw (dot.center) to [out=\neangle, in=down] ([xshift=1.4cm] f.south) to (2.4,1.7) node [above] {$S$};
\begin{pgfonlayer}{foreground}
\node (prho) [anchor=south, smallbox, thick] at ([xshift=-1.4cm] r2.south) {$P_\rho$};
\end{pgfonlayer}
\draw (r2.60) to [out=up, in=up, looseness=1.2] ([xshift=0.5cm] r2.north -| prho.south) to ([xshift=0.5cm, yshift=-0.0cm] prho.south) to [out=down, in=down, looseness=2] (prho.south) to (prho.north |- r2.north) to [out=up, in=down] (r3.-60);
\end{tikzpicture}
\end{aligned}
\nonumber
\end{equation}
\begin{equation}
\quad=\quad
\sum_\rho c_\rho
\begin{aligned}
\begin{tikzpicture}[thick, xscale=-1, xscale=\tempxscale, scale=\licsscale, scale=1]
\draw (1,-1.5) node (f) [smallbox] {$f$} to (1,-0.7) node (r2) [smallbox] {$\rho$};
\draw (1, 0.9) node (r3) [smallbox] {$\rho ^\dag$};
\draw (r2.125) to [out=up, in=down, in looseness=0.7] (r3.-125);
\draw (r3.north) to (1,1.7) node [above] {$G$};
\draw (f.south) to [out=down, in=\nwangle] +(0.7,-1) node [smallcircle] (dot) {} to +(0,-1) node [smallcircle] {} node [anchor=east, smallfont] {\makebox[0pt][r]{$\displaystyle \frac{1}{\sqrt{|S|}}$}};
\draw (dot.center) to [out=\neangle, in=down] ([xshift=1.4cm] f.south) to (2.4,1.7) node [above] {$S$};
\begin{pgfonlayer}{foreground}
\node (prho) [anchor=center, smallbox, thick, minimum width=0pt, inner sep=1pt] at ($(r3.-60)!0.5!(r2.60)$) {$P_\rho$};
\end{pgfonlayer}
\draw (r2.60) to [out=up, in=down] (prho.south) to (prho.north) to [out=up, in=down] (r3.-60);
\end{tikzpicture}
\end{aligned}
\end{equation}
We are now ready to consider the measurement stage of the protocol, which is a projective measurement onto the even superposition state of $S$.

Suppose that $f$ is $P_\rho$-constant. Then using equation~\eqref{eq:genbalancedconstant} we obtain the following expression for the overall state of the joint system:
\begin{equation}
\sum_\rho c_\rho
\begin{aligned}
\begin{tikzpicture}[thick, xscale=-1, xscale=\tempxscale, scale=\licsscale]
\node (1,-1.5) (f) [smallcircle] {};
\node [smallbox] (rhodag) at (0,2.0) {$\rho ^\dag$};
\draw (0, 0.9) node (r3) [smallbox] {$x_\rho$};
\draw (f.south) to [out=down, in=\nwangle] +(0.7,-1) node [smallcircle] (dot) {} to +(0,-1) node [smallcircle, smallfont] {} node [anchor=east, smallfont] {\makebox[0pt][r]{$\displaystyle \frac{1}{\sqrt{|S|}}$}};
\draw (dot) to [out=\neangle, in=down] (1.4,0) to (1.4,2.75) node [above] {$S$};
\draw (rhodag.north) to (0,2.75) node [above] {$G$};
\draw (r3.60) to (rhodag.-60);
\draw (r3.120) to (rhodag.-120);
\end{tikzpicture}
\end{aligned}
\quad=\quad
{\sum_\rho \frac {c_\rho} {\sqrt{|S|}}}
\begin{aligned}
\begin{tikzpicture}[thick, xscale=-1, xscale=\tempxscale, scale=\licsscale]
\draw (0,0) node [smallcircle] {} to (0,2) node [above] {$S$};
\draw (-1.5,0) node (rdag) [smallbox] {$\rho ^\dag$} to (-1.5,2) node [above] {$G$};
\node [smallbox] (x) at (-1.5,-1.2) {$x_\rho$};
\draw (x.60) to (rdag.-60);
\draw (x.120) to (rdag.-120);
\node [white] at (-1,-2.8) {{\makebox[0pt][r]{$\displaystyle \frac{1}{\sqrt{|S|}}$}}};
\end{tikzpicture}
\end{aligned}
\end{equation}
This is a product state, with the system $S$ in the even superposition state. It is not the zero state, since the maps $\rho ^\dag$ are injective, and at least one $x ^\rho$ is nonzero. A projective measurement of $S$ onto the even superposition state is therefore guaranteed to be successful.

Finally, suppose that $f$ is $P_\rho$-balanced. To determine the effect of a projective measurement onto the even superposition state, we compose~\eqref{eq:gendjafterunitary} with a partial isometry onto that subspace of the system $S$:
\begin{equation}
\sum_\rho c_\rho
\hspace{-25pt}
\begin{aligned}
\begin{tikzpicture}[thick, xscale=-1, xscale=\tempxscale, scale=\licsscale]
\draw (1,-1.5) node (f) [smallbox] {$f$} to (1,-0.7) node (r2) [smallbox] {$\rho$};
\draw (1, 0.9) node (r3) [smallbox] {$\rho ^\dag$};
\draw (r2.125) to [out=up, in=down, in looseness=0.7] (r3.-125);
\draw (r3.north) to (1,1.7) node [above] {$G$};
\draw (f.south) to [out=down, in=\nwangle] +(0.7,-1) node [smallcircle] (dot) {} to +(0,-1) node [smallcircle] {} node [anchor=east, smallfont] {\makebox[0pt][r]{$\displaystyle \frac{1}{\sqrt{|S|}}$}};
\draw (dot) to [out=\neangle, in=down] ([xshift=1.4cm] f.south) to (2.4,1.0) node [smallcircle] {} node [anchor=east, smallfont] {$\displaystyle \frac{1}{\sqrt{|S|}}$};
\begin{pgfonlayer}{foreground}
\node (prho) [anchor=center, smallbox, thick, inner sep=1pt, minimum width=0pt] at ($(r3.-60)!0.5!(r2.60)$) {$P_\rho$};
\end{pgfonlayer}
\draw (r2.60) to [out=up, in=down] (prho.south) to (prho.north) to [out=up, in=down] (r3.-60);
\end{tikzpicture}
\end{aligned}
\quad=\quad
\sum_\rho \frac{c_\rho} {|S|}
\begin{aligned}
\begin{tikzpicture}[thick, xscale=-1, xscale=\tempxscale, scale=\licsscale]
\draw (1,-1.5) node (f) [smallbox] {$f$} to (1,-0.7) node (r2) [smallbox] {$\rho$};
\draw (1, 0.9) node (r3) [smallbox] {$\rho ^\dag$};
\draw (r2.125) to [out=up, in=down, in looseness=0.7] (r3.-125);
\draw (r3.north) to (1,1.7) node [above] {$G$};
\draw (f.south) to +(0,-1) node [smallcircle] {};
\begin{pgfonlayer}{foreground}
\node (prho) [smallbox, thick, minimum width=0pt, inner sep=1pt] at ($(r3.-60)!0.5!(r2.60)$) {$P_\rho$};
\end{pgfonlayer}
\draw (r2.60) to [out=up, in=down] (prho.south) to (prho.north) to [out=up, in=down] (r3.-60);
\node at (0.5,-3.80)  [anchor=east, white] {\makebox[0pt][r]{$\displaystyle \frac{1}{\sqrt{|S|}}$}};
\end{tikzpicture}
\end{aligned}
\quad=\quad
0
\end{equation}
Using the definition~\eqref{eq:genbalancedconstant} of a $P_\rho$-balanced function, this composite is zero, so we have shown that in the $P_\rho$-balanced case there is no possibility of measuring the system $S$ in the even superposition state.

\section{The Single-Shot Grover Algorithm}
\label{sec:grover}

\subsection{Introduction}

\noindent
In this section we give a topological analysis of the single-shot Grover algorithm, which is used to search a set $S$ for marked elements, defined in terms of an indicator function
\[
S \sxto f \{ 0,1 \}
\]
onto the group of integers under addition modulo 2. Grover's algorithm is nondeterministic in general, which makes it difficult to treat in our formalism. We restrict attention to the single-shot case in which a single iteration of the algorithm is guaranteed to return a marked element. This corresponds to the case that exactly $\frac 1 4$ of the elements are marked. Our topological analysis makes it clear why Grover's algorithm is always successful in this case. 

We then explore a generalization of Grover's algorithm, made natural by our approach, in which our set $S$ is equipped with a function to an arbitrary finite group. We demonstrate that this generalized algorithm is guaranteed to return an \emph{unbalanced} element of $S$ after a single iteration. Given an irreducible representation $\rho$ of $G$, we define the element $s \in S$ to be \emph{balanced} if the following holds:
\begin{equation}
\rho(f(s)) = \frac 2 {|S|} \sum_{t \in S} \rho(f(t))
\end{equation}
In words, an element $s \in S$ is balanced if $\rho ( f (s ))$  equals twice the average value, using the uniform measure on $S$. If an element is not balanced, then it is \textit{unbalanced}.

To illustrate this, consider the case of $G=\Z _3$, and write $\underline 6$ for the set $\{1, 2, 3, 4, 5, 6 \}$. Consider the function $\underline 6 \sxto f \Z_3$ defined as follows: 
\begin{align*}
f(1) &= 0 & f(5) &= 1 & f(6)&=2
\\
f(2) &= 0
\\
f(3) &= 0
\\
f(4) &= 0
\end{align*}
We choose the representation $\Z_3 \sxto \rho \C$ with $\rho(1) = e ^{2 \pi i / 3}$. Under this representation, the average value of the quantity $\rho(f(n))$ over all $n \in [6]$ is
\begin{equation}
\frac{1}{6}(4 + e ^{2 \pi i/3} + e ^{-2 \pi i/3}) = \frac{1}{2}.
\end{equation}
As a result, the elements 1, 2, 3 and 4 are balanced, since for these values $\rho(f(n))$ takes the value 1 which is twice the average value. The elements 5 and 6 are unbalanced, so these are the only possible measurement outcomes, and the algorithm will determine one of them in a single query.

This works for any function $f: S \to \Z_3$, for any set $S$, as long as it takes the values of $\Z_3$ in the ratio $4:1:1$ in some order. In this general case, as with the simpler scenario we examined above, the algorithm will return an element of $S$ marked with one of the less-frequent values of $\Z_3$ in a single query.

We emphasize that this cannot be handled by the standard Grover algorithm. If we knew beforehand what values of $G$ are taken by the unbalanced elements of $S$ (in our simple example, 1 and 2), and if we knew the fraction of elements which are unbalanced (in this case $\frac 1 3$), then we could compose $f$ with a function to $\Z_2$ that takes $0 \mapsto 0$, $1 \mapsto 1$ and $2 \mapsto 1$, and apply an optimal variant of Grover's conventional search algorithm to find a rare element with certainty after a single iteration~\cite{ck99-qds}. However, in general, this information may not be known.

\subsection{Topological presentation}

\noindent
Our topological perspective demonstrates why the single-shot Grover algorithm is successful, when exactly a quarter of the elements are marked. The overall topological structure of the algorithm is as follows:
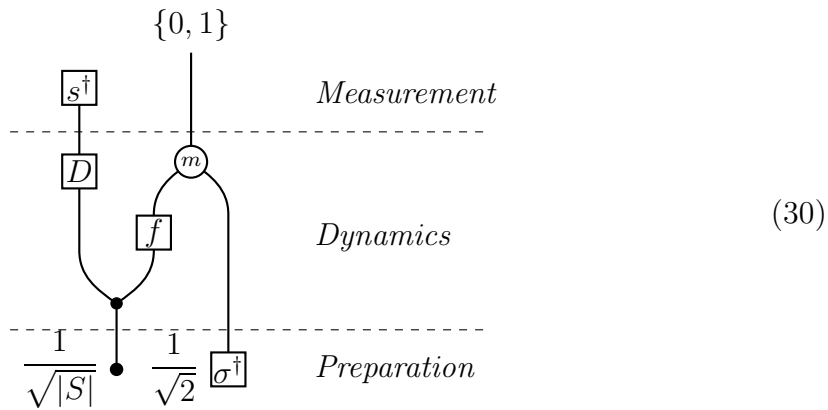
\begin{equation}
\label{eq:grovertopological}
\begin{aligned}
\begin{tikzpicture}[thick, scale=\licsscale]
\begin{pgfonlayer}{foreground}
    \node (dot) [smallcircle] at (0,1) {};
    \node (f) [smallbox, anchor=south, thick] at (0.7,2) {$f$};
    \node (m) [circlelabel, thick] at ([xshift=0.7cm, yshift=1cm] f.north) {$m$};
\end{pgfonlayer}
\draw (0,-0.25)
        node [smallcircle] (bdot) {}
    to (0,1)
    to [out=\nwangle, in=south] (-0.7,2)
    to ([yshift=1.4cm] m.center -| -0.7,1)
        node (rho) [smallbox] {$s ^\dag$};
\draw (0,1)
    to [out=\neangle, in=south] (f.south)
    to (f.north)
    to [out=up, in=\swangle] +(0.7,1)
    to [out=\seangle, in=up] +(0.7,-1)
    to (2.1,-0.25)
        node [smallbox] (sigmadag) {$\sigma ^\dag$};
\draw (m.center) to (1.4,5.75)
        node [above] {$\{0,1\}$};
\node [anchor=east] at (bdot.west) {$\displaystyle \frac{1}{\sqrt{|S|}}$};
\node [anchor=east] at (sigmadag.west) {$\displaystyle \frac{1}{\sqrt{2}}$};
\node [smallbox] at (-0.7,3.5) {$D$};
\node at (3.5,-0.25) [anchor=west] {\textit{Preparation}};
\draw [sepstyle] (-2,0.5) to (7,0.5);
\node at (3.5,2.25) [anchor=west] {\textit{Dynamics}};
\draw [sepstyle] (-2,4.25) to (7,4.25);
\node at (3.5,5) [anchor=west] {\textit{Measurement}};
\end{tikzpicture}
\end{aligned}
\end{equation}
This represents the final state after a successful projective measurement of the first system in the basis state $\ket s$.  Furthermore, we can decompose $D$ in the following way, as a linear combination of two other diagrams:
\begin{equation}
\label{eq:difftopological}
\begin{aligned}
\begin{tikzpicture}[thick, scale=\licsscale]
\draw (0,0) node [below] {$S$} to node [smallbox] {$D$} (0,2) node [above] {$S$};
\end{tikzpicture}
\end{aligned}
\hspace{5pt}\quad=\quad\hspace{5pt}
\hspace{-5pt}
\begin{aligned}
\begin{tikzpicture}[thick, scale=\licsscale]
\draw (0,0) node [below] {$S$} to (0,2) node [above] {$S$};
\end{tikzpicture}
\end{aligned}
\quad-\quad
\frac{2}{|S|}
\begin{aligned}
\begin{tikzpicture}[thick, scale=\licsscale]
\draw (0,0) node [below] {$S$} to (0,0.7) node [smallcircle] {};
\draw (0,2) node [above] {$S$} to (0,1.3) node [smallcircle] {};
\end{tikzpicture}
\end{aligned}
\end{equation}
This linear decomposition allows us to perform a completely topological analysis of the algorithm. By equation~\eqref{eq:copyonleg}, the linear map $\sigma ^\dag$ is duplicated by the group multiplication vertex $m$ in equation~\eqref{eq:grovertopological}, allowing us to rewrite it in the following way:
\begin{equation}
\begin{aligned}
\begin{tikzpicture}[thick, scale=\licsscale, scale=0.8]
\begin{pgfonlayer}{foreground}
    \node (dot) [smallcircle] at (0,1) {};
    \node (f) [smallbox, anchor=south, thick] at (0.7,2) {$f$};
    \node (m) [smallbox, thick] at (2.8,3.5) {$\sigma ^\dag$};
\end{pgfonlayer}
\draw (0,-0.25)
        node [smallcircle] (bdot) {}
    to (0,1)
    to [out=\nwangle, in=south] (-0.7,2)
    to ([yshift=1.4cm] m.center -| -0.7,1)
        node (rho) [smallbox] {$s ^\dag$};
\draw (0,1)
    to [out=\neangle, in=south] (f.south)
    to (f.north)
        to (0.7,3.5)
        node [smallbox] {$\sigma$};
\draw (m.center) to (2.8,5.75)
        node [above] {$\{0,1\}$};
\node [smallbox] at (-0.7,3.5) {$D$};
\node [anchor=east, smallfont] at ([xshift=4pt] m.west) {$\displaystyle \frac{1}{\sqrt{2}}$};
\node [anchor=east] at (bdot.west) {$\displaystyle \frac{1}{\sqrt{|S|}}$};
\end{tikzpicture}
\end{aligned}
\end{equation}
Since this is a product state, we can neglect the second system, rewriting the state of the first system using the topological decomposition~\eqref{eq:difftopological} of the diffusion operator $D$:
\begin{align}
\nonumber
{\textstyle\frac{1}{\sqrt{|S|}}}
\begin{aligned}
\begin{tikzpicture}[thick, scale=\licsscale, scale=0.8]
\begin{pgfonlayer}{foreground}
    \node (dot) [smallcircle] at (0,1) {};
    \node (f) [smallbox, anchor=south, thick] at (0.7,2) {$f$};
\end{pgfonlayer}
\draw (0,-0.25)
        node [smallcircle] (bdot) {}
    to (0,1)
    to [out=\nwangle, in=south] (-0.7,2)
    to (-0.7,3.5)
        node (rho) [smallbox] {$s ^\dag$};
\draw (0,1)
    to [out=\neangle, in=south] (f.south)
    to (f.north)
        to (0.7,3.5)
        node [smallbox] {$\sigma$};
\node [smallbox, anchor=south] at (-0.7,2) {$D$};
\end{tikzpicture}
\end{aligned}
\quad
&=
\quad
{\textstyle\frac{1}{\sqrt{|S|}}}
\left(\,\,
\begin{aligned}
\begin{tikzpicture}[thick, scale=\licsscale, scale=0.8]
\begin{pgfonlayer}{foreground}
    \node (dot) [smallcircle] at (0,1) {};
    \node (f) [smallbox, anchor=south, thick] at (0.7,2) {$f$};
\end{pgfonlayer}
\draw (0,-0.25)
        node [smallcircle] (bdot) {}
    to (0,1)
    to [out=\nwangle, in=south] (-0.7,2)
    to ([yshift=-0.0cm] m.center -| -0.7,1)
        node (rho) [smallbox] {$s ^\dag$};
\draw (0,1)
    to [out=\neangle, in=south] (f.south)
    to (f.north)
        to (0.7,3.5)
        node [smallbox] {$\sigma$};
\end{tikzpicture}
\end{aligned}
\hspace{3pt}-\hspace{3pt}
{\textstyle\frac{2}{|S|}}
\begin{aligned}
\begin{tikzpicture}[thick, scale=\licsscale, scale=0.8]
\begin{pgfonlayer}{foreground}
    \node (dot) [smallcircle] at (0,1) {};
    \node (f) [smallbox, anchor=south, thick] at (0.7,2) {$f$};
\end{pgfonlayer}
\draw (0,-0.25)
        node [smallcircle] (bdot) {}
    to (0,1)
    to [out=\nwangle, in=south] (-0.7,2) node [smallcircle] {};
\draw (-0.7,2.7) node [smallcircle] {}
    to ([yshift=0.00cm] m.center -| -0.7,1)
        node (rho) [smallbox] {$s ^\dag$};
\draw (0,1)
    to [out=\neangle, in=south] (f.south)
    to (f.north)
        to (0.7,3.5)
        node [smallbox] {$\sigma$};
\end{tikzpicture}
\end{aligned}
\,\,
\right)
\\
&= \quad
{\textstyle\frac{1}{\sqrt{|S|}}}
\left(
\,\,
\begin{aligned}
\begin{tikzpicture}[thick, scale=\licsscale, scale=0.8]
\draw (0,0) node [smallbox] {$s$} to (0,1.25) node [smallbox] {$f$} to (0,2.5) node [smallbox] {$\sigma$};
\end{tikzpicture}
\end{aligned}
\quad-\quad
{\textstyle\frac{2}{|S|}}
\begin{aligned}
\begin{tikzpicture}[thick, scale=\licsscale, scale=0.8]
\draw (0,0) node [smallcircle] {} to (0,1.25) node [smallbox] {$f$} to (0,2.5) node [smallbox] {$\sigma$};
\node [smallbox, draw=white, fill=none] at (0,0) {};
\draw (0,0) to (0,0.5);
\end{tikzpicture}
\end{aligned}
\,\,\right)
\end{align}
In the final equality we use properties of the topological algebra of finite sets, as described in Sections~\ref{sec:finitesets} and \ref{sec:functionsbetweensets}.

This is zero for those $s \in S$ satisfying the following equation:
\begin{equation}
\label{eq:zerocondition}
\begin{aligned}
\begin{tikzpicture}[thick, scale=\licsscale, scale=0.8]
\draw (0,0) node [smallbox] {$s$} to (0,1.25) node [smallbox] {$f$} to (0,2.5) node [smallbox] {$\sigma$};
\end{tikzpicture}
\end{aligned}
\quad=\quad
\frac{2}{|S|}
\begin{aligned}
\begin{tikzpicture}[thick, scale=\licsscale, scale=0.8]
\draw (0,0) node [smallcircle] {} to (0,1.25) node [smallbox] {$f$} to (0,2.5) node [smallbox] {$\sigma$};
\node [smallbox, draw=white, fill=none] at (0,0) {};
\draw (0,0) to (0,0.5);
\end{tikzpicture}
\end{aligned}
\end{equation}
Such an $s$ cannot be the result of a measurement. The right-hand side of this expression represents twice the average value of the function $f$ under the representation $\sigma$.
The left-hand side represents the value of the function $f$ at the element $s$, under the representation $\sigma$. If the right-hand side is equal to ${+}1$, it is straightforward to see that $f$ must take the value $1$ on three-quarters of the elements of $S$. If the right-hand side is equal to ${-}1$, then $f$ must take the value 1 on one-quarter of the elements of $S$.

This recovers the standard result from Grover search: to find a marked element with certainty after a single step, exactly a quarter of elements must be marked. Note that standard Grover theory assumes that the marked elements take the value 1 under the function $f$, but this is not required: the marked elements could just as well take the value 0, and the ummarked elements the value 1. It is clear that these two cases give rise to essentially the same procedure, since the difference corresponds to an overall phase. So one can argue that the traditional Grover algorithm should not be thought of as a search algorithm, but as a `frequency detection' algorithm, an idea which fits well with our generalized versions.

\subsection{Generalization}

\noindent
We generalize this by replacing the group $\{0,1\}$ with an arbitrary finite group $G$. Loosely, the generalized algorithm is guarantees to find an element $s \in S$ for which $\rho(f(s))$ does not take twice its average value, restricted to some chosen subspace of the representation space. The following diagram gives a topological outline of the algorithm:
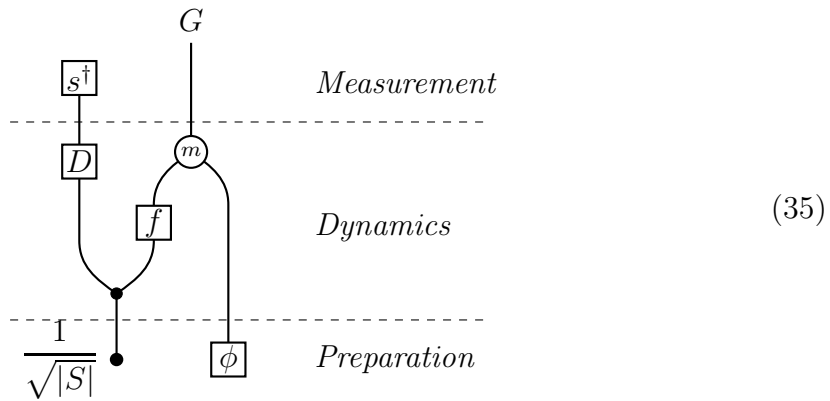
\begin{equation}
\begin{aligned}
\begin{tikzpicture}[thick, scale=\licsscale]
\begin{pgfonlayer}{foreground}
    \node (dot) [smallcircle] at (0,1) {};
    \node (f) [smallbox, anchor=south, thick] at (0.7,2) {$f$};
    \node (m) [circlelabel, thick] at ([xshift=0.7cm, yshift=1cm] f.north) {$m$};
\end{pgfonlayer}
\draw (0,-0.25)
        node [smallcircle] (bdot) {}
    to (0,1)
    to [out=\nwangle, in=south] (-0.7,2)
    to ([yshift=1.4cm] m.center -| -0.7,1)
        node (rho) [smallbox] {$s ^\dag$};
\draw (0,1)
    to [out=\neangle, in=south] (f.south)
    to (f.north)
    to [out=up, in=\swangle] +(0.7,1)
    to [out=\seangle, in=up] +(0.7,-1)
    to (2.1,-0.25)
        node [smallbox] (sigmadag) {$\phi$};
\draw (m.center) to (1.4,5.75)
        node [above] {$G$};
\node [anchor=east] at (bdot.west) {$\displaystyle \frac{1}{\sqrt{|S|}}$};
\node [smallbox] at (-0.7,3.5) {$D$};
\node at (3.5,-0.25) [anchor=west] {\textit{Preparation}};
\draw [sepstyle] (-2,0.5) to (7,0.5);
\node at (3.5,2.25) [anchor=west] {\textit{Dynamics}};
\draw [sepstyle] (-2,4.25) to (7,4.25);
\node at (3.5,5) [anchor=west] {\textit{Measurement}};
\end{tikzpicture}
\end{aligned}
\end{equation}
The system $G$ is initialized in a state $\phi$ which is chosen in the following way, just as with the generalized Deutsch-Jozsa algorithm presented in Section~\ref{sec:gendeutschjozsa}:
\begin{equation}
\label{eq:phidef}
\begin{aligned}
\begin{tikzpicture}[thick, scale=\licsscale]
\begin{pgfonlayer}{foreground}
\draw [thick] (2.1,0.5) node [above] {$G$} to (2.1,-0.75)
    node [smallbox] {$\phi$};
\end{pgfonlayer}
\node at (2.1,-0.25) [smallbox, minimum width=17pt, white] (rhodag) {$\phi$};
\node [white] (prho) at ([yshift=-0.3cm] rhodag.-55) [anchor=north, smallbox, minimum width=13pt] {$P_\rho$};
\draw [white] ([yshift=-0.0cm] prho.south) to +(0,-0.25) to [out=down, in=down, looseness=2] ([yshift=-0.25cm] prho.south -| rhodag.-125);
\end{tikzpicture}
\end{aligned}
\quad:=\quad
\sum _\rho c_\rho \,\, 
\begin{aligned}
\begin{tikzpicture}[thick, scale=\licsscale]
\draw (2.1,0.5) node [above] {$G$} to (2.1,-0.25)
    node [smallbox, minimum width=17pt] (rhodag) {$\rho ^\dag$};
\node (prho) at ([yshift=-0.3cm] rhodag.-55) [anchor=north, smallbox, minimum width=13pt] {$P_\rho$};
\draw (rhodag.-55) to (prho.north);
\draw (prho.south) to +(0,-0.25) to [out=down, in=down, looseness=2] ([yshift=-0.25cm] prho.south -| rhodag.-125) to (rhodag.-125);
\end{tikzpicture}
\end{aligned}
\end{equation}

By following a similar argument to the previous section, we obtain the following expression for the basis elements $\ket s \in S$ which have zero amplitude to be measured:
\begin{equation}
\label{eq:generalizedzerocondition}
\begin{aligned}
\begin{tikzpicture}[thick, scale=\licsscale]
\draw (1.4,1.25) node [smallbox, anchor=north] {$\phi$} to [out=up, in=\seangle] (0.7,2.4);
\draw (0,0) node [smallbox, anchor=north] {$s$} to (0,1.25) node [smallbox, anchor=north] {$f$} to [out=up, in=\swangle] (0.7,2.4) node (m) [circlelabel] {$m$} to (0.7,3.4);
\end{tikzpicture}
\end{aligned}
\quad=\quad
\frac{2}{|S|}\,\,
\begin{aligned}
\begin{tikzpicture}[thick, scale=\licsscale]
\draw (1.4,1.25) node [smallbox, anchor=north] {$\phi$} to [out=up, in=\seangle] (0.7,2.4);
\draw (0,-0.2) node [smallcircle] {} to (0,1.25) node [smallbox, anchor=north] {$f$} to [out=up, in=\swangle] (0.7,2.4) node (m) [circlelabel] {$m$} to (0.7,3.4);
\node [smallbox, anchor=north, white] at (0.7,0.0) {};
\end{tikzpicture}
\end{aligned}
\end{equation}
From this equation, we see that the algorithm will never return those elements $s \in S$ such that, restricted the support of $m(-,\phi)$, the group element $f(s)$ takes twice its average value. We can apply definition~\eqref{eq:phidef} to see how $m(-,\phi)$ acts in terms of our chosen projectors $P_\rho$:
\def\tempxscale{1}
\begin{equation}
\sum_\rho c_\rho
\begin{aligned}
\begin{tikzpicture}[thick, xscale=\tempxscale, scale=\licsscale]
\begin{pgfonlayer}{foreground}
    \node (m) [circlelabel, thick] at (1.4,5) {$m$};
\end{pgfonlayer}
\draw (0.7,2) node [below] {$G$} to (0.7,4)
    to [out=up, in=\swangle] (m.center)
    to [out=\seangle, in=up] +(0.7,-1)
    to (2.1,4)
        node [smallbox, minimum width=17pt, anchor=north] (rhodag) {$\rho ^\dag$};
\node (prho) at ([yshift=-0.3cm] rhodag.-55) [anchor=north, smallbox, minimum width=13pt] {$P_\rho$};
\draw (m.center) to (1.4,6)
        node [above] {$G$};
\draw (rhodag.-55) to (prho.north);
\draw (prho.south) to [out=down, in=down, looseness=2] (prho.south -| rhodag.-125) to (rhodag.-125);
\end{tikzpicture}
\end{aligned}
\hspace{-5pt}
\quad=\quad
\sum_\rho c_\rho
\begin{aligned}
\begin{tikzpicture}[thick, xscale=-1, xscale=\tempxscale, scale=\licsscale]
\node (r2) [smallbox] at (1,-0.4) {$\rho$};
\draw (0.3, 0.9) node (r3) [smallbox] {$\rho ^\dag$};
\draw (r2.120) to [out=up, in=down, in looseness=0.7] (r3.-120);
\draw (r3.north) to (0.3,2) node [above] {$G$};
\draw (r2.south) to (1,-2) node [below] {$G$};
\begin{pgfonlayer}{foreground}
\node (prho) [anchor=south, smallbox, thick] at ([xshift=-1.4cm] r2.south) {$P_\rho$};
\end{pgfonlayer}
\draw (r2.60) to [out=up, in=up, looseness=1.2] ([xshift=0.5cm] r2.north -| prho.south) to ([xshift=0.5cm, yshift=-0.0cm] prho.south) to [out=down, in=down, looseness=2] (prho.south) to (prho.north |- r2.north) to [out=up, in=down] (r3.-60);
\end{tikzpicture}
\end{aligned}
\hspace{-5pt}
\quad=\quad
\sum_\rho c_\rho
\begin{aligned}
\begin{tikzpicture}[thick, xscale=-1, xscale=\tempxscale, scale=\licsscale]
\node (r2) [smallbox] at (1,-1.3) {$\rho$};
\draw (1, 0.9) node (r3) [smallbox] {$\rho ^\dag$};
\draw (r2.120) to [out=up, in=down, in looseness=0.7] (r3.-120);
\draw (r3.north) to (1,1.8) node [above] {$G$};
\draw (r2.south) to (1,-2.2) node [below] {$G$};
\begin{pgfonlayer}{foreground}
\node (prho) [anchor=center, smallbox, thick, minimum width=00pt, inner sep=1pt] at (r2.60 |- (0,-0.2) {$P_\rho$};
\end{pgfonlayer}
\draw (r2.60) to [out=up, in=down] (prho.south) to (prho.north) to [out=up, in=down] (r3.-60);
\end{tikzpicture}
\end{aligned}
\end{equation}
The support of this operator is the union of subspaces defined by our projectors $P_\rho$, across all the equivalence classes of irreducible representations.

Applying this to equation~\eqref{eq:generalizedzerocondition}, we see that a basis element $\ket s \in S$ cannot be the result of the measurement if it satisfies the following condition, which generalizes condition~\eqref{eq:zerocondition} above for a balanced element of $S$:
\begin{equation}
\forall \rho \quad\quad
\begin{aligned}
\begin{tikzpicture}[thick, xscale=-1, xscale=\tempxscale, scale=\licsscale]
\node (r2) [smallbox] at (1,-1.3) {$\rho$};
\draw (r2.120) to (r2.120 |- 0,0.7);
\draw (r2.south) to +(0,-1) node [smallbox, anchor=south] {$f$} to +(0,-2) node [smallbox, anchor=south] {$s$};
\begin{pgfonlayer}{foreground}
\node (prho) [anchor=south, smallbox, thick, minimum width=0pt, inner sep=1pt] at ([yshift=0.3cm] r2.60) {$P_\rho$};
\end{pgfonlayer}
\draw (r2.60) to [out=up, in=down] (prho.south) to (prho.north) to (prho.north |- 0,0.7);
\node [above] at (1,0.7) {\makebox[0pt]{$\Mat(\dim (\rho))$}};
\end{tikzpicture}
\end{aligned}
\quad=\quad
\frac{2}{|S|} \,\,
\begin{aligned}
\begin{tikzpicture}[thick, xscale=-1, xscale=\tempxscale, scale=\licsscale]
\node [smallbox, white] at (1,-3.32) {$s$};
\node (r2) [smallbox] at (1,-1.3) {$\rho$};
\draw (r2.120) to (r2.120 |- 0,0.7);
\draw (r2.south) to +(0,-1) node [smallbox, anchor=south] {$f$} to +(0,-1.7) node [smallcircle, anchor=south] {};
\begin{pgfonlayer}{foreground}
\node (prho) [anchor=south, smallbox, thick, minimum width=0pt, inner sep=1pt] at ([yshift=0.3cm] r2.60) {$P_\rho$};
\end{pgfonlayer}
\draw (r2.60) to [out=up, in=down] (prho.south) to (prho.north) to (prho.north |- 0,0.7);
\node [above] at (1,0.7) {\makebox[0pt]{$\Mat(\dim (\rho))$}};
\end{tikzpicture}
\end{aligned}
\end{equation}
This says that, under each representation $\rho$, restricted to the subspace $\id _{\dim(\rho)} \otimes P_\rho$, the group element $f(s)$ takes twice its average value.

\section{The hidden subgroup algorithm}
\label{sec:hiddensubgroup}

\subsection{Introduction}

\noindent
The hidden subgroup family of algorithms is rich, containing Deutsch's original algorithm, as well as Simon's algorithm and Shor's algorithm. We use our topological notation to prove correctness of the algorithm, and clarify its structure.

We are given a function $G \sxto f X$, promised to be constant on the cosets of some normal subgroup $H \subseteq G$, and distinct otherwise. This says exactly that the function $f$ factorizes as
\begin{equation}
\label{eq:hsfactorization}
\begin{aligned}
\begin{tikzpicture}
\node [anchor=east] (S) at (0,0) {$G$};
\node (1) at (1.2,0) {$G/H$};
\node [anchor=west] (G) at (2.4,0) {$S\,,$};
\draw [->>] (S) to node [auto, swap] {$q$} (1);
\draw [>->] (1) to node [auto, swap] {$s$} (G);
\draw [->] (S) to [out=45, in=135] node [auto] {$f$} (G);
\end{tikzpicture}
\end{aligned}
\end{equation}
where $q$ is a surjective projection onto the quotient group $G/H$, and $s$ is an embedding into some set. The algorithm determines the subgroup $H$ in $O(\log|G|)$ trials.

Traditionally applied only in the case where $G$ is abelian, it was extended to the case of normal subgroups of arbitrary finite groups by Hallgren, Russell and Ta-shma~\cite{hrt00-nsr}. It is this generalized version that we treat here. Our approach gives a clean separation between the group-theoretical and quantum aspects of the protocol, and yields a proof of correctness which is simpler, shorter, and easier to follow.

\subsection{Topological version}

\noindent
The following diagram summarizes our topological account of the algorithm.
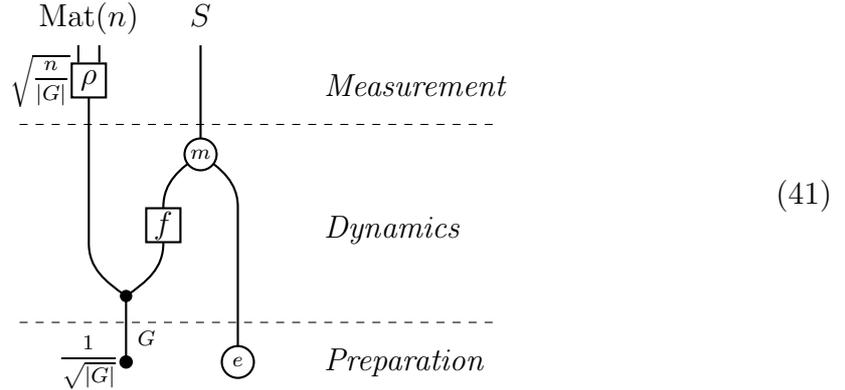
\begin{equation}
\begin{aligned}
\begin{tikzpicture}[thick, scale=\licsscale]
\begin{pgfonlayer}{foreground}
    \node (dot) [smallcircle] at (0,1) {};
    \node (f) [smallbox, anchor=south, thick] at (0.7,2) {$f$};
    \node (m) [circlelabel, thick] at ([xshift=0.7cm, yshift=1cm] f.north) {$m$};
\end{pgfonlayer}
\draw (0,-0.25)
        node [smallcircle] (dot2) {}
    to (0,1)
    to [out=\nwangle, in=south] (-0.7,2)
    to ([yshift=1.4cm] m.center -| -0.7,1)
        node (rho) [smallbox] {$\rho$};
\node [above] at (-0.7,5.75) {$\Mat(n)$};
\draw (0,1)
    to [out=\neangle, in=south] (f.south)
    to (f.north)
    to [out=up, in=\swangle] +(0.7,1)
    to [out=\seangle, in=up] +(0.7,-1)
    to (2.1,-0.25)
        node [circlelabel] {$e$};
\draw (m.center) to (1.4,5.75)
        node [above] {$S\vphantom($};
\node [anchor=east, smallfont, inner sep=0pt] at (dot2.west) {$\displaystyle \frac{1}{\sqrt{|G|}}$};
\node at (3.5,-0.25) [anchor=west] {\textit{Preparation}};
\draw [sepstyle] (-2,0.5) to (7,0.5);
\node at (3.5,2.25) [anchor=west] {\textit{Dynamics}};
\draw [sepstyle] (-2,4.25) to (7,4.25);
\node at (3.5,5) [anchor=west] {\textit{Measurement}};
\node at (rho.west) [anchor=east, smallfont, inner sep=0pt] {$\displaystyle \sqrt{\frac{n}{|G|}}$};
\draw (rho.60) to (rho.60 |- 0,5.75);
\draw (rho.120) to (rho.120 |- 0,5.75);
\node [smallfont, anchor=west] at (0,0.2) {$G$};
\end{tikzpicture}
\end{aligned}
\end{equation}
We begin by preparing two systems: one with state space given by a group algebra $G$ in the uniform superposition state; and another with state space given by a set $S$ in some basis state $\ket e$. We then perform the unitary dynamics described by the second section, which as described in Section~\ref{sec:buildingtheunitary} recreates the action of the unitary operator $U_f$ in the conventional presentation. For this purpose we require a group structure on $S$; we choose any structure (such as a cyclic group structure) such that $e$ is the identity element.

The procedure finishes with a measurement on the first system in the partition defined by the irreducible representations, as described in Section~\ref{sec:repprojmeas}. Using a topological encoding of a result from group theory, we demonstrate that the only representations that can be successfully measured are those which arise as restrictions of representations of the quotient group $G/H$. Classical postprocessing then suffices to deduce the subgroup $H$ in $O(\log|G|)$ trials~\cite{hrt00-nsr}.

After the unitary dynamics step, the systems are in the following state:
\begin{equation}
\label{eq:hsafterunitary}
\hspace{-60pt}
\begin{aligned}
\begin{tikzpicture}[thick, scale=\licsscale, scale=0.8]
\begin{pgfonlayer}{foreground}
    \node (dot) [smallcircle] at (0,1) {};
    \node (f) [smallbox, anchor=south, thick] at (0.7,2) {$f$};
    \node (m) [circlelabel, thick] at ([xshift=0.7cm, yshift=1cm] f.north) {$m$};
\end{pgfonlayer}
\draw (0,0.25)
        node [smallcircle] (dot2) {}
    to (0,1)
    to [out=\nwangle, in=south] (-0.7,2)
    to (-0.7, 4.65)
        node [above] {$G$};
\draw (0,1)
    to [out=\neangle, in=south] (f.south)
    to (f.north)
    to [out=up, in=\swangle] +(0.7,1)
    to [out=\seangle, in=up] +(0.7,-1)
    to (2.1,0.25)
        node [circlelabel] {$e$};
\draw (m.center) to (1.4,4.65)
        node [above] {$S$};
\node [anchor=east, smallfont, inner sep=0pt] at (dot2.west) {$\displaystyle \frac{1}{\sqrt{|G|}}$};
\end{tikzpicture}
\end{aligned}
\hspace{5pt}=\hspace{-3pt}
\begin{aligned}
\begin{tikzpicture}[thick, scale=\licsscale, scale=0.8]
\begin{pgfonlayer}{foreground}
    \node (dot) [smallcircle] at (0,1) {};
    \node (f) [smallbox, anchor=south, thick] at (0.7,2.5) {$f$};
\end{pgfonlayer}
\draw (0,0.25)
        node (dot2) [smallcircle] {}
    to (0,1)
    to [out=\nwangle, in=south] (-0.7,2)
    to (-0.7, 4.65)
        node [above] {$G$};
\draw (0,1)
    to [out=\neangle, in=south] (0.7,2)
    to (f.north)
    to (0.7,4.65) node [above] {$S$};
\node [anchor=east, smallfont, inner sep=0pt] at (dot2.west) {$\displaystyle \frac{1}{\sqrt{|G|}}$};
\end{tikzpicture}
\end{aligned}
\hspace{5pt}=\hspace{-3pt}
\begin{aligned}
\begin{tikzpicture}[thick, scale=\licsscale, scale=0.8]
\begin{pgfonlayer}{foreground}
    \node (dot) [smallcircle] at (0,1) {};
    \node (pi) [smallbox, anchor=south, thick] at (0.7,2) {$q$};
    \node (s) [smallbox, anchor=south, thick] at (0.7,3.2) {$s$};
\end{pgfonlayer}
\draw (0,0.25)
        node (dot2) [smallcircle] {}
    to (0,1)
    to [out=\nwangle, in=south] (-0.7,2)
    to (-0.7, 4.65)
        node [above] {$G$};
\draw (0,1)
    to [out=\neangle, in=south] (pi.south)
    to (0.7,4.65) node [above] {$S$};
\node [anchor=east, smallfont, inner sep=0pt] at (dot2.west) {$\displaystyle \frac{1}{\sqrt{|G|}}$};
\node at (0.7,2.88) [smallfont, anchor=west] {};
\end{tikzpicture}
\end{aligned}
\hspace{-50pt}
\end{equation}
These are topological representations of the following algebraic expression:
\begin{equation}
\frac{1}{\sqrt{|G|}} \sum _{g \in G} \ket g \otimes \ket {f(g)}
\end{equation}
On the right-hand side of~\eqref{eq:hsafterunitary} we have rewritten $f$ in terms of its promised factorization~\eqref{eq:hsfactorization} via an unknown quotient group~$G/H$.

\subsection{Performing the measurement}

\noindent
We now describe measurement of the first system $G$, using a projective measurement determined by the irreducible representations. To understand what results are possible, we compose the state~\eqref{eq:hsafterunitary} with the linear map
\begin{equation}
\begin{aligned}
\begin{tikzpicture}[thick, scale=\licsscale]
\draw (0,0) node [below] {$G$} to (0,1);
\node at (0,2) [above] {$\Mat(n)$};
\node (tau) [smallbox] at (0,1) {$\rho$};
\node [anchor=east, smallfont] at (tau.west) {$\displaystyle \sqrt{\frac{n}{|G|}}$};
\draw (tau.60) to (tau.60 |- 0,2);
\draw (tau.120) to (tau.120 |- 0,2);
\end{tikzpicture}
\end{aligned}
\end{equation}
on the first factor, which as discussed in Section~\ref{sec:repprojmeas} is a partial isometry projecting onto the subspace of the group algebra corresponding to the irreducible representation~$\rho$. This gives the following result:
\begin{equation}
\label{eq:hsprojectedstate}
\frac{\sqrt{n}}{|G|}
\hspace{-10pt}
\begin{aligned}
\begin{tikzpicture}[thick, scale=\licsscale, scale=0.8]
\begin{pgfonlayer}{foreground}
    \node (dot) [smallcircle] at (0,1) {};
    \node (pi) [smallbox, anchor=south, thick] at (0.7,2) {$q$};
    \node (s) [smallbox, anchor=south, thick] at (0.7,3) {$s$};
\end{pgfonlayer}
\draw (0,0.25)
        node (dot2) [smallcircle] {}
    to (0,1)
    to [out=\nwangle, in=south] (-0.7,2)
    to (-0.7, 4.0);
\node [above] at (-0.7,5.2) {$\Mat(n)$};
\draw (0,1)
    to [out=\neangle, in=south] (pi.south)
    to (0.7,5.2) node [above] {$S\vphantom{(}$};
\node (rho) at (-0.7,4) [smallbox, anchor=south] {$\rho$};
\draw (rho.60) to (rho.60 |- 0,5.2);
\draw (rho.120) to (rho.120 |- 0,5.2);
\end{tikzpicture}
\end{aligned}
\quad=\quad
\frac{\sqrt{n}}{|G|}
\hspace{-5pt}
\begin{aligned}
\begin{tikzpicture}[thick, scale=\licsscale, scale=0.8]
\begin{pgfonlayer}{foreground}
    \node (dot) [smallcircle] at (0,1) {};
    \node (pi) [smallbox, anchor=south, thick] at (0.7,2) {$q$};
    \node (s) [smallbox, anchor=south, thick] at (0.7,3) {$s$};
\end{pgfonlayer}
\draw (0,0.25)
        node (dot2) [smallcircle] {}
    to (0,1)
    to [out=\nwangle, in=south] (-0.7,2);
\node at (-0.7,5.2) [above] {$\Mat(n)$};
\draw (0,1)
    to [out=\neangle, in=south] (pi.south)
    to (0.7,5.2) node [above] {$S\vphantom{(}$};
\node (rho) at (-0.7,2) [smallbox, anchor=south] {$\rho$};
\draw (rho.60) to (rho.60 |- 0,5.2);
\draw (rho.120) to (rho.120 |- 0,5.2);
\end{tikzpicture}
\end{aligned}
\end{equation}
By Theorem~\ref{thm:graphicalnormalsubgroup}, this is zero exactly when $\rho$ does not factor through the unknown quotient group $G/H$. So the only irreducible representations that can be measured in this way are those which factorize as $G \sxto q G/H  \sxto \tau \Mat(n)$, for some irreducible representation $\tau$ of $G/H$. In this case, our projected state~\eqref{eq:hsprojectedstate} can be rewritten as follows, making use of the topological properties of functions between basis elements and even surjections as presented in Section~\ref{sec:functionsbetweensets}:
\begin{equation}
\hspace{-60pt}
{\textstyle\frac{\sqrt{n}}{|G|}}
\begin{aligned}
\begin{tikzpicture}[thick, scale=\licsscale]
\begin{pgfonlayer}{foreground}
    \node (dot) [smallcircle] at (0,1) {};
    \node (pi) [smallbox, anchor=south, thick] at (0.7,2) {$q$};
    \node (s) [smallbox, anchor=south, thick] at (0.7,3.2) {$s$};
\end{pgfonlayer}
\draw (0,0.25)
        node (dot2) [smallcircle] {}
    to (0,1)
    to [out=\nwangle, in=south] (-0.7,2)
    to (-0.7, 3.55);
\node at (-0.7, 4.45) [above] {\makebox[0pt]{$\Mat(n)$}};
\draw (0,1)
    to [out=\neangle, in=south] (pi.south)
    to (0.7,4.45) node [above] {$S\vphantom{(}$};
\node at (-0.7,2.25) [smallbox] {$q$};
\node (tau) at (-0.7,3.2) [smallbox, anchor=south] {$\tau$};
\node at (-0.75,2.90) [smallfont, anchor=west] {$G/H$};
\node at (0.65,2.90) [smallfont, anchor=west] {\makebox[0pt][l]{$G/H$}};
\node at (0.0,0.625) [smallfont, anchor=west] {\makebox[0pt][l]{$G$}};
\draw (tau.60) to (tau.60 |- 0,4.45);
\draw (tau.120) to (tau.120 |- 0,4.45);
\end{tikzpicture}
\end{aligned}
\quad=\hspace{3pt}
{{\textstyle\frac{\sqrt{n}}{|G|}}
\begin{aligned}
\begin{tikzpicture}[thick, scale=\licsscale]
\begin{pgfonlayer}{foreground}
    \node (dot) [smallcircle] at (0,1.75) {};
    \node (s) [smallbox, anchor=south, thick] at (0.7,3.2) {$s$};
\end{pgfonlayer}
\draw (0,0.25)
        node (dot2) [smallcircle] {}
    to (dot.center)
    to [out=\nwangle, in=south] +(-0.7,1)
    to (-0.7, 3.45);
\node at (-0.7, 4.45) [above] {\makebox[0pt]{$\Mat(n)$}};
\draw (dot.center)
    to [out=\neangle, in=south] +(0.7,1)
    to (0.7,4.45) node [above] {$S\vphantom{(}$};
\node (tau) at (-0.7,3.2) [smallbox, anchor=south] {$\tau$};
\node at (0,1.0) [smallbox] {$q$};
\node at (0.0,0.5) [smallfont, anchor=west] {\makebox[0pt][l]{$G$}};
\node at (0.0,1.5) [smallfont, anchor=west] {\makebox[0pt][l]{$G/H$}};
\draw (tau.60) to (tau.60 |- 0,4.45);
\draw (tau.120) to (tau.120 |- 0,4.45);
\end{tikzpicture}
\end{aligned}
\hspace{0pt}=\hspace{3pt}
{\textstyle{\scriptstyle\sqrt{n}} \frac{|H|}{|G|}}}
\begin{aligned}
\begin{tikzpicture}[thick, scale=\licsscale]
\begin{pgfonlayer}{foreground}
    \node (dot) [smallcircle] at (0,1.75) {};
    \node (s) [smallbox, anchor=south, thick] at (0.7,3.2) {$s$};
\end{pgfonlayer}
\draw (0,0.75)
        node (dot2) [smallcircle] {}
    to (dot.center)
    to [out=\nwangle, in=south] +(-0.7,1)
    to (-0.7, 3.35);
\node at (-0.7, 4.45) [above] {\makebox[0pt]{$\Mat(n)$}};
\draw (dot.center)
    to [out=\neangle, in=south] +(0.7,1)
    to (0.7,4.45) node [above] {$S\vphantom{(}$};
\node (tau) at (-0.7,3.2) [smallbox, anchor=south] {$\tau$};
\node [smallcircle, fill=white, white] at (0,0.25) {};
\node at (0.0,1.25) [smallfont, anchor=west] {\makebox[0pt][l]{$G/H$}};
\draw (tau.60) to (tau.60 |- 0,4.45);
\draw (tau.120) to (tau.120 |- 0,4.45);
\end{tikzpicture}
\end{aligned}
\hspace{-50pt}
\end{equation}
We calculate the norm of this vector in the following way, employing the graphical characterization of injective functions given in Section~\ref{sec:functionsbetweensets}, and both the cyclic property of the trace and the topological expression of the dimension of a vector space as developed in Section~\ref{sec:trace}:
\allowdisplaybreaks[1]
{\begin{calign}
\nonumber
n \frac{|H| ^2}{|G| ^2}
\hspace{4pt}
\begin{aligned}
\begin{tikzpicture}[thick, scale=\licsscale, xscale=0.8, yscale=0.9]
\draw (0,0.25) node [smallcircle] {} to (0,1) node [smallcircle] {} to [out=\nwangle, in=down] (-0.7,2) node (t) [anchor=south, smallbox] {$\tau$};
\draw (-0.7,3.75) node (tdag) [anchor=north, smallbox] {$\tau ^\dag$} to [out=up, in=\swangle] (0,4.75) node [smallcircle] {} to (0,5.5) node [smallcircle] {};
\draw (0,1) to [out=\neangle, in=down] (0.7,2) node [anchor=south, smallbox] {$s$} to (0.7,3.75) node [anchor=north, smallbox] {$s ^\dag$} to [out=up, in=\seangle] (0,4.75);
\node at (0,0.625) [anchor=west, smallfont] {$G/H$};
\node at (0,5.12) [anchor=west, smallfont] {$G/H$};
\draw (t.60) to (tdag.-60);
\draw (t.120) to (tdag.-120);
\end{tikzpicture}
\end{aligned}
\quad=\quad
n \frac{|H| ^2}{|G| ^2}
\hspace{4pt}
\begin{aligned}
\begin{tikzpicture}[thick, scale=\licsscale, xscale=0.8, yscale=0.9]
\draw (0,0.25) node [smallcircle] {} to (0,1) node [smallcircle] {} to [out=\nwangle, in=down] (-0.7,2) node (t) [anchor=south, smallbox] {$\tau$};
\draw (-0.7,3.75) node (tdag) [anchor=north, smallbox] {$\tau ^\dag$} to [out=up, in=\swangle] (0,4.75) node [smallcircle] {} to (0,5.5) node [smallcircle] {};
\draw (0,1) to [out=\neangle, in=down] (0.7,2) to (0.7,3.75) to [out=up, in=\seangle] (0,4.75);
\node at (0,0.625) [anchor=west, smallfont] {$G/H$};
\node at (0,5.12) [anchor=west, smallfont] {$G/H$};
\draw (t.60) to (tdag.-60);
\draw (t.120) to (tdag.-120);
\end{tikzpicture}
\end{aligned}
\\
=\hspace{5pt}
n \frac{|H| ^2}{|G| ^2}
\hspace{0pt}
\begin{aligned}
\begin{tikzpicture}[thick, scale=\licsscale, xscale=0.8, yscale=0.9]
\draw (0,0.25) node [smallcircle] {} to (0,1) node [smallcircle] {} to [out=\nwangle, in=down] (-0.7,2) node [anchor=south, smallbox, white] {$\tau ^\dag$};
\draw (-0.7,3.75) node [anchor=north, smallbox, white] {$\tau$} to [out=up, in=\swangle] (0,4.75) node [smallcircle] {} to (0,5.5) node [smallcircle] {};
\draw (0,1) to [out=\neangle, in=down] (0.7,2) to (0.7,3.75) to [out=up, in=\seangle] (0,4.75);
\begin{scope}[xshift=10pt]
\draw (0,0.25) node [smallcircle] {} to (0,1) node [smallcircle] {} to [out=\nwangle, in=down] (-0.7,2);
\draw (-0.7,3.75) to [out=up, in=\swangle] (0,4.75) node [smallcircle] {} to (0,5.5) node [smallcircle] {};
\draw (0,1) to [out=\neangle, in=down] (0.7,2) to (0.7,3.75) to [out=up, in=\seangle] (0,4.75);
\node at (0,0.625) [anchor=west, smallfont] {$\Mat(n)$};
\node at (0,5.12) [anchor=west, smallfont] {$\Mat(n)$};
\end{scope}
\node (tdag) [anchor=south, smallbox] at ([xshift=5pt] -0.7,2) {$\tau ^\dag$};
\node (t) [anchor=north, smallbox] at ([xshift=5pt] -0.7,3.75) {$\tau$};
\draw (t.south) to (tdag.north);
\end{tikzpicture}
\end{aligned}
\hspace{-13pt}=\hspace{5pt}
 \frac{|H|}{|G|}
\hspace{4pt}
\begin{aligned}
\begin{tikzpicture}[thick, smallfont, scale=\licsscale, xscale=0.8, yscale=0.9]
\draw (0,0.25) node [smallcircle] {} to (0,1) node [smallcircle] {} to [out=\nwangle, in=down] (-0.7,2) to (-0.7,3.75) to [out=up, in=\swangle] (0,4.75) node [smallcircle] {} to (0,5.5) node [smallcircle] {};
\draw (0,1) to [out=\neangle, in=down] (0.7,2) to (0.7,3.75) to [out=up, in=\seangle] (0,4.75);
\begin{scope}[xshift=10pt]
\draw (0,0.25) node [smallcircle] {} to (0,1) node [smallcircle] {} to [out=\nwangle, in=down] (-0.7,2) to (-0.7,3.75) to [out=up, in=\swangle] (0,4.75) node [smallcircle] {} to (0,5.5) node [smallcircle] {};
\draw (0,1) to [out=\neangle, in=down] (0.7,2) to (0.7,3.75) to [out=up, in=\seangle] (0,4.75);
\node at (0,0.625) [anchor=west, smallfont] {$\Mat(n)$};
\node at (0,5.12) [anchor=west, smallfont] {$\Mat(n)$};
\end{scope}
\end{tikzpicture}
\end{aligned}
\hspace{-13pt}=\hspace{5pt}
\frac{|H|}{|G|} n^2
\end{calign}}So the measurement will return exactly those irreducible representations of $G$ that factor through the unknown quotient group $G/H$, with a probability proportional to the square of the dimension of the representation.

\enlargethispage{-1.8in}
\bibliographystyle{plain}
\bibliography{../../../jov}

\begin{thebibliography}{10}

\bibitem{ac04-csqp}
Samson Abramsky and Bob Coecke.
\newblock A categorical semantics of quantum protocols.
\newblock {\em Proceedings of the 19th Annual IEEE Symposium on Logic in
  Computer Science}, pages 415--425, 2004.
\newblock IEEE Computer Science Press.

\bibitem{ac08-cqm}
Samson Abramsky and Bob Coecke.
\newblock {\em Handbook of Quantum Logic and Quantum Structures}, volume~2,
  chapter Categorical Quantum Mechanics.
\newblock Elsevier, 2008.

\bibitem{bs10-rosetta}
John~C Baez and Mike Stay.
\newblock {\em New Structures for Physics}, chapter Physics, Topology, Logic
  and Computation: A Rosetta Stone, pages 95--172.
\newblock Springer, 2010.

\bibitem{cst09-cta}
T.~Ceccherini-Silberstein, F.~Scarabotti, and F.~Tolli.
\newblock Clifford theory and applications.
\newblock {\em Journal of Mathematical Sciences}, 156(1):29--43, 2009.

\bibitem{ck99-qds}
Dong~Pyo Chi and Jinsoo Kim.
\newblock Quantum database search by a single query.
\newblock {\em Lecture Notes in Computer Science}, 1509:148--151, 1999.

\bibitem{c03-loe}
Bob Coecke.
\newblock The logic of entanglement: An invitation.
\newblock Technical report, University of Oxford, 2003.
\newblock Computing Laboratory Research Report PRG-RR-03-12.

\bibitem{cd11-iqo}
Bob Coecke and Ross Duncan.
\newblock Interacting quantum observables: Categorical algebra and
  diagrammatics.
\newblock {\em New Journal of Physics}, 13, 2011.

\bibitem{cp06-pnwt}
Bob Coecke and Eric~Oliver Paquette.
\newblock {POVM}s and {N}aimark's theorem without sums.
\newblock 2006.

\bibitem{cpp09-cqs}
Bob Coecke, Eric~Oliver Paquette, and Dusko Pavlovic.
\newblock Classical and quantum structuralism.
\newblock In Simon Gay and Ian Mackie, editors, {\em Semantic Techniques in
  Quantum Computation}, pages 29--69. Cambridge University Press, 2010.

\bibitem{cp06-qmws}
Bob Coecke and Dusko Pavlovic.
\newblock {\em The Mathematics of Quantum Computation and Technology}, chapter
  Quantum Measurements Without Sums.
\newblock Taylor and Francis, 2006.

\bibitem{cpv08-dfb}
Bob Coecke, Dusko Pavlovic, and Jamie Vicary.
\newblock Commutative dagger-{F}robenius algebras in {FdHilb} are orthogonal
  bases.
\newblock ({RR}-08-03), 2008.
\newblock Technical Report.

\bibitem{cp08-ecc}
Bob Coecke and Simon Perdrix.
\newblock Environment and classical channels in categorical quantum mechanics.
\newblock In {\em Proceedings of the 19th {EACSL} Annual Conference on Computer
  Science Logic}, number 6247 in Lecture Notes in Computer Science, 2010.

\bibitem{hrt00-nsr}
Sean Hallgren, Alexander Russell, and Amnon Ta-shma.
\newblock Normal subgroup reconstruction and quantum computation using group
  representations.
\newblock In {\em Proceedings of the 32nd ACM Symposium on Theory of
  Computing}, 2000.

\bibitem{h99-coqa}
Peter H{\o}yer.
\newblock Conjugated operators in quantum algorithms.
\newblock {\em Physical Review A}, 59:3280--3289, 1999.

\bibitem{kl80-cccc}
Gregory~M. Kelly and Miguel~L. Laplaza.
\newblock Coherence for compact closed categories.
\newblock {\em Journal of Pure and Applied Algebra}, 19:193--213, 1980.

\bibitem{k04-fa2d}
Joachim Kock.
\newblock {\em Frobenius Algebras and {2D} Topological Quantum Field Theories}.
\newblock Cambridge University Press, 2004.

\bibitem{bdb06-etp}
Andrew J.~Duncan Michael~Batty and Samuel~L. Braunstein.
\newblock Extending the promise of the {D}eutsch-{J}ozsa-{H}{\o}yer algorithm
  for finite groups.
\newblock {\em Journal of Computation and Mathematics}, 9:40--63, 2006.

\bibitem{nc-qcqi}
Michael Nielsen and Isaac Chuang.
\newblock {\em Quantum Computation and Quantum Information}.
\newblock Cambridge University Press, 2000.

\bibitem{p10-gaqc}
Dusko Pavlovic.
\newblock Geometry of abstraction in quantum computation.
\newblock {\em arXiv:1006.1010}, 2010.

\bibitem{p71-andt}
Roger Penrose.
\newblock Applications of negative-dimensional tensors.
\newblock In D.~J.~A. Welsh, editor, {\em Combinatorial Mathematics and its
  Applications}. Academic Press, 1971.

\bibitem{rsw04-gcsa}
R.~Rosebrugh, N.~Sabadini, and R.~F.~C. Walters.
\newblock Generic commutative separable algebras and cospans of graphs.
\newblock {\em Theory and Applications of Categories}, 15(164--177):6, 2005.

\bibitem{s11-sgl}
Peter Selinger.
\newblock {\em New Structures for Physics}, chapter A Survey of Graphical
  Languages for Monoidal Categories, pages 289--355.
\newblock Number 813 in Lecture Notes in Physics. Springer, 2011.

\bibitem{v12-hsqp}
Jamie Vicary.
\newblock Higher semantics of quantum protocols.
\newblock {\em Proceedings of the 27th Annual IEEE Symposium on Logic in
  Computer Science}, 2012.
\newblock IEEE Computer Science Press, to appear.

\end{thebibliography}

\newpage
\appendix

\section{Topological algebra}
\label{sec:topologicalalgebra}

\subsection{Introduction}

\noindent
Here we present a topological notation for linear algebra, and show how it gives a powerful way to display various properties of finite groups and finite sets. Built on early work of Penrose~\cite{p71-andt}, these techniques have now found wide use in quantum foundations and quantum information~\cite{bs10-rosetta,c03-loe,k04-fa2d}. A mathematical foundation for the notation is given by category theory~\cite{kl80-cccc,s11-sgl}, but this will remain entirely beneath the surface here. This material is suitable for an appendix since all of it can be found in other places in the literature, but in only a scattered fashion.

We begin by defining the notation, then demonstrate how we can use it to represent theorems about finite sets and finite groups in a topological way. These will form the heart of our topological proofs of correctness for the Deutsch-Jozsa, hidden subgroup and single-shot Grover algorithms.

\subsection{Linear algebra}

In the topological notation, diagrams represent linear maps between finite-dimensional Hilbert spaces. A vertically-oriented wire represents the identity map on a finite-dimensional Hilbert space, which we will often label with its name:
\begin{equation}
\begin{aligned}
\begin{tikzpicture}[thick, scale=\licsscale]
\draw (0,0) node [below] {$H$} to (0,1.5);
\end{tikzpicture}
\end{aligned}
\end{equation}
Nontrivial linear maps are represented by boxes, or sometimes simply by a vertex, with its domain represented by the input wires at the bottom, and its codomain represented by output wires at the top. So we draw the following diagram to represent a linear map of type $H \sxto p J$:
\begin{equation}
\begin{aligned}
\begin{tikzpicture}[thick, scale=\licsscale]
\node (f) at (0, 0) [smallbox] {$p$};
\draw (f.south -| 0,0) to (0,-0.75) node [below] {$H$};
\draw (f.north -| 0,0) to (0,+0.75) node [above] {$J$};
\end{tikzpicture}
\end{aligned}
\end{equation}
Horizontal juxtaposition of diagrams represents tensor product of linear maps, and vertical juxtaposition represents composition of linear maps. So, for example, given maps $H \otimes J \sxto q K$ and $M \sxto r J \otimes L$
\begin{calign}
\begin{aligned}
\begin{tikzpicture}[thick, scale=\licsscale]
\node (f) at (0, 0) [smallbox, minimum width=1.1cm, minimum height = 0.45cm] {$q$};
\draw (f.south -| -0.5, 0) to +(0,-0.55) node [below] {$H$};
\draw (f.south -| 0.5, 0) to +(0,-0.55) node [below] {$J$};
\draw (f.north -| 0,0) to +(0,0.55) node [above] {$K$};
\end{tikzpicture}
\end{aligned}
&
\begin{aligned}
\begin{tikzpicture}[thick, scale=\licsscale]
\node (f) at (0, 0) [smallbox, minimum width=1.1cm, minimum height = 0.45cm] {$r$};
\draw (f.north -| -0.5, 0) to +(0,0.55) node [above] {$J$};
\draw (f.north -| 0.5, 0) to +(0,0.55) node [above] {$L$};
\draw (f.south -| 0,0) to +(0,-0.55) node [below] {$M$};
\end{tikzpicture}
\end{aligned}
\end{calign}
we represent the composite $(q \otimes \id_L) \circ (\id_H \otimes r)$ with the following diagram:
\begin{equation}
\begin{aligned}
\begin{tikzpicture}[thick, scale=\licsscale]
\node (g) at (0, 0) [smallbox, minimum width=1.1cm, minimum height = 0.45cm] {$q$};
\node (h) at ([yshift=-0.75cm] g.south -| 1,0) [anchor=north, smallbox, minimum width=1.1cm, minimum height = 0.45cm] {$r$};
\draw (0.5,0 |- g.south) to node [right] {$J$} (0.5,0 |- h.north);
\draw (0,0 |- g.north) to +(0,0.55) node [above] {$K$};
\draw (1.5,0 |- h.north) to ([yshift=0.55cm] g.north -| 1.5,0) node [above] {$L$};
\draw (1,0 |- h.south) to +(0,-0.55) node [below] {$M$};
\draw (-0.5,0 |- g.south) to ([yshift=-0.55cm] h.south -| -0.5,0) node [below] {$H$};
\end{tikzpicture}
\end{aligned}
\end{equation}
The identity on the 1\-dimensional Hilbert space \C{} is represented as the empty diagram:
\begin{equation}
\begin{aligned}\rule{0pt}{30pt}\end{aligned}
\end{equation}
Juxtaposing this with any diagram leaves the diagram unchanged, which is consistent with its mathematical role as the unit for the tensor product operation.

Nontrivial algebraic interactions between the tensor product and composite of linear maps are made transparent by the formalism. In particular, for a family of maps $H \sxto s J$, $J \sxto t K$, $L \sxto u M$ and $M \sxto v N$, the equal algebraic composites $(t \otimes v) \circ (s \otimes u)$ and $(t \circ s) \otimes (v \circ u)$ have the same graphical representation:
\begin{equation}
\begin{aligned}
\begin{tikzpicture}[thick, scale=\licsscale]
\draw (0,0) node [below] {$H$} to (0,0.55) node (s) [anchor=south, smallbox] {$s$} to (s.north) to node [auto] {$J$} +(0,0.75) node (t) [anchor=south, smallbox] {$t$};
\draw (t.north) to +(0,0.55) node [above] {$K$};
\begin{scope}[xshift=1.5cm]
\draw (0,0) node [below] {$L$} to (0,0.55) node (s) [anchor=south, smallbox] {$u$} to (s.north) to node [auto, swap] {$M$} +(0,0.75) node (t) [anchor=south, smallbox] {$v$};
\draw (t.north) to +(0,0.55) node [above] {$N$};
\end{scope}
\end{tikzpicture}
\end{aligned}
\end{equation}
The formalism also allows us to change the relative heights of boxes and move components around, as long as the connectivity between the different boxes is maintained, without changing the value of the diagram. For example, the following diagrams represent equal linear maps:
\begin{equation}
\begin{aligned}
\begin{tikzpicture}[thick, scale=\licsscale, scale=0.8]
\node (w) [smallbox] at (-0.7,1.5) {$w$};
\node (x) [smallbox] at (0.7,0.5) {$x$};
\draw (0,-1) node [below] {$H$} to [out=up, in=down] (-0.7,0.5) to (w.south);
\draw (0,3) node [above] {$J$} to [out=down, in=up] (0.7,1.5) to (x.north);
\end{tikzpicture}
\end{aligned}
\quad=\quad
\begin{aligned}
\begin{tikzpicture}[thick, scale=\licsscale, scale=0.8]
\node (w) [smallbox] at (0,0.2) {$w$};
\node (x) [smallbox] at (0,1.8) {$x$};
\draw (0,-1) node [below] {$H$} to (w.south);
\draw (0,3) node [above] {$J$} to (x.north);
\end{tikzpicture}
\end{aligned}
\quad=\quad
\begin{aligned}
\begin{tikzpicture}[thick, xscale=-1, scale=\licsscale, scale=0.8]
\node (w) [smallbox] at (-0.7,1.5) {$w$};
\node (x) [smallbox] at (0.7,0.5) {$x$};
\draw (0,-1) node [below] {$H$} to [out=up, in=down] (-0.7,0.5) to (w.south);
\draw (0,3) node [above] {$J$} to [out=down, in=up] (0.7,1.5) to (x.north);
\end{tikzpicture}
\end{aligned}
\end{equation}

\subsection{Finite sets}
\label{sec:finitesets}

\noindent
Given a finite set $S$, we can form the free complex vector space~$\C[S]$, with a canonical basis in bijection with the elements of $S$. We equip this with an inner product such that the basis elements are orthonormal, giving rise to a Hilbert space. To reduce the weight of our notation we allow the symbol $S$ to also stand for this Hilbert space, and write $\ket s$ to represent the basis element corresponding to the element $s$ of the original set. The vector space \mbox{S} carries a commutative algebra structure, defined by the maps \mbox{$S \otimes S \sxto m S$} and \mbox{$\C \sxto u S$}:
\begin{align}
\label{eq:basismultiply}
m \big( \ket i \otimes \ket j \big) :=& \delta_{i,j} \ket i
\\
\label{eq:basisunit}
u :=& \textstyle \sum _i \ket i
\end{align}
We represent these graphically in the following way.
\begin{calign}
\begin{aligned}
\begin{tikzpicture}[thick, scale=\licsscale]
\draw (-0.7,-1) to [out=up, in=\swangle] (0,0);
\draw (0.7,-1) to [out=up, in=\seangle] (0,0);
\draw (0,0) to (0,1);
\node (m) at (0,0) [smallcircle] {};
\end{tikzpicture}
\end{aligned}
&
\begin{aligned}
\begin{tikzpicture}[thick, scale=\licsscale]
\draw [white] (0,-1) to (0,1);
\draw (0,0) to (0,1);
\node (m) at (0,0) [smallcircle] {};
\end{tikzpicture}
\end{aligned}
\\
\nonumber
S \otimes S \sxto m S
&
\C \sxto u S
\end{calign}
Since these are linear maps between Hilbert spaces we can also construct their adjoints, which we denote as follows:
\begin{calign}
\label{eq:copydelete}
\begin{aligned}
\begin{tikzpicture}[thick, yscale=-1, scale=\licsscale]
\draw (-0.7,-1) to [out=up, in=\swangle] (0,0);
\draw (0.7,-1) to [out=up, in=\seangle] (0,0);
\draw (0,0) to (0,1);
\node (m) at (0,0) [smallcircle] {};
\end{tikzpicture}
\end{aligned}
&
\begin{aligned}
\begin{tikzpicture}[thick, yscale=-1, scale=\licsscale]
\draw [white] (0,-1) to (0,1);
\draw (0,0) to (0,1);
\node (m) at (0,0) [smallcircle] {};
\end{tikzpicture}
\end{aligned}
\\
\nonumber
S \sxto {m ^\dag} S \otimes S
&
S \sxto {u^\dag} \C
\end{calign}
The maps $m ^\dag$ and $u^\dag$ provide a canonical way to `copy' and `erase' elements of $S$, in a way which is governed by the canonical basis.

The relationships between composites of these are governed by a strong theorem~\cite{cpv08-dfb,k04-fa2d,rsw04-gcsa}.
\begin{theorem}
\label{thm:spider}
Given a commutative algebra on a finite-dimensional Hilbert space $H$, the following are equivalent:
\begin{enumerate}
\item Any two composites of the multiplication, the unit and their adjoints, having the same inputs and outputs, are equal iff they have the same connectivity.
\item The algebra arises in the manner of equations \eqref{eq:basismultiply} and \eqref{eq:basisunit} from an orthonormal basis of $H$.
\end{enumerate}
\end{theorem}

\noindent
In particular, the following equalities are implied by these equivalent properties. Together these form the definition of a \emph{classical structure}, also known as a \emph{special \dag-Frobenius algebra}:
\def\frobscale{0.5}
\begin{calign}
\begin{aligned}
\begin{tikzpicture}[thick, scale=\frobscale]
\draw (0,0) to [out=up, in=\swangle] (0.7, 1);
\draw (1.4,0) to [out=up, in=\seangle] (0.7,1);
\draw (2.8,0) to [out=up, in=\seangle] (1.55,2);
\draw (0.7,1) to [out=up, in=\swangle] (1.55, 2);
\draw (1.55,2) to (1.55, 2.72);
\node [smallcircle] at (0.7,1) {};
\node [smallcircle] at (1.55,2) {};
\end{tikzpicture}
\end{aligned}
\quad=\quad
\begin{aligned}
\begin{tikzpicture}[xscale=-1, thick, scale=\frobscale]
\draw (0,0) to [out=up, in=\swangle] (0.7, 1);
\draw (1.4,0) to [out=up, in=\seangle] (0.7,1);
\draw (2.8,0) to [out=up, in=\seangle] (1.55,2);
\draw (0.7,1) to [out=up, in=\swangle] (1.55, 2);
\draw (1.55,2) to (1.55, 2.72);
\node [smallcircle] at (0.7,1) {};
\node [smallcircle] at (1.55,2) {};
\end{tikzpicture}
\end{aligned}
&
\begin{aligned}
\begin{tikzpicture}[thick, scale=\frobscale]
\draw (0,0) to [out=up, in=down] (1.4,1.2) to [out=up, in=\seangle] (0.7,2) to (0.7,2.75);
\draw (1.4,0) to [out=up, in=down] (0,1.2) to [out=up ,in=\swangle] (0.7,2);
\node [smallcircle] at (0.7,2) {};
\end{tikzpicture}
\end{aligned}
\quad=\quad
\begin{aligned}
\begin{tikzpicture}[thick, scale=\frobscale]
\draw (0,0) to [out=up, in=\swangle] (0.7,2) to (0.7,2.75);
\draw (1.4,0) to [out=up, in=\seangle] (0.7,2) node [smallcircle] {};
\end{tikzpicture}
\end{aligned}
\end{calign}

\vspace{-10pt}
\begin{calign}
\begin{aligned}
\begin{tikzpicture}[thick, scale=\frobscale]
\draw (0,-0.75) to (0,0) to [out=up, in=\swangle] (0.7,1);
\draw (1.4,0) to [out=up, in=\seangle] (0.7,1);
\draw (0.7,1) to (0.7,1.75);
\node [smallcircle] at (0.7,1) {};
\node [smallcircle] at (1.4,0) {};
\end{tikzpicture}
\end{aligned}
\quad=\quad
\begin{aligned}
\begin{tikzpicture}[thick, scale=\frobscale]
\draw (0,0) to (0,2.5);
\end{tikzpicture}
\end{aligned}
\quad=\quad
\begin{aligned}
\begin{tikzpicture}[thick, xscale=-1, scale=\frobscale]
\draw (0,-0.75) to (0,0) to [out=up, in=\swangle] (0.7,1);
\draw (1.4,0) to [out=up, in=\seangle] (0.7,1);
\draw (0.7,1) to (0.7,1.75);
\node [smallcircle] at (0.7,1) {};
\node [smallcircle] at (1.4,0) {};
\end{tikzpicture}
\end{aligned}
&
\begin{aligned}
\begin{tikzpicture}[thick, scale=\frobscale]
\draw (0,0) to (0,0.65) node [smallcircle] {} to (0,0.65) to [out=\nwangle, in=down] (-0.5,1.25) to [out=up, in=\swangle] (0,1.85) node [smallcircle] {} to (0,2.5);
\draw (0,0.65) to [out=\neangle, in=down] (0.5,1.25) to [out=up, in=\seangle] (0,1.85);
\end{tikzpicture}
\end{aligned}
\quad=\quad
\begin{aligned}
\begin{tikzpicture}[thick, scale=\frobscale]
\draw (0,0) to (0,2.5);
\end{tikzpicture}
\end{aligned}
\end{calign}

\vspace{-10pt}
\begin{calign}
\begin{aligned}
\begin{tikzpicture}[thick, scale=\frobscale]
\draw (0,0) to (0,0.75) node [smallcircle] {};
\draw (0,0.75) to [out=\neangle, in=\swangle] (0.7,1.75) node [smallcircle] {};
\draw (0,0.75) to [out=\nwangle, in=down] (-0.7,2.5);
\draw (0.7,1.75) to (0.7,2.5);
\draw (0.7,1.75) to [out=\seangle, in=up] (1.4,0);
\end{tikzpicture}
\end{aligned}
\quad=\quad
\begin{aligned}
\begin{tikzpicture}[thick, scale=\frobscale]
\draw (0,0) to [out=up, in=\swangle] (0.7,0.75) node [smallcircle] {};
\draw (1.4,0) to [out=up, in=\seangle] (0.7,0.75) node [smallcircle] {};
\draw (0,2.5) to [out=down, in=\nwangle] (0.7,1.75) node [smallcircle] {};
\draw (1.4,2.5) to [out=down, in=\neangle] (0.7,1.75) node [smallcircle] {};
\draw (0.7,0.75) to (0.7,1.75);
\end{tikzpicture}
\end{aligned}
\quad=\quad
\begin{aligned}
\begin{tikzpicture}[thick, yscale=-1, scale=\frobscale]
\draw (0,0) to (0,0.75) node [smallcircle] {};
\draw (0,0.75) to [out=\neangle, in=\swangle] (0.7,1.75) node [smallcircle] {};
\draw (0,0.75) to [out=\nwangle, in=down] (-0.7,2.5);
\draw (0.7,1.75) to (0.7,2.5);
\draw (0.7,1.75) to [out=\seangle, in=up] (1.4,0);
\end{tikzpicture}
\end{aligned}
\end{calign}
Conversely, these classical structure axioms imply property 1 of Theorem~\ref{thm:spider}, a result known as the Spider Theorem~\cite{cp06-pnwt}.

\subsection{Functions between sets}
\label{sec:functionsbetweensets}

\noindent
A function $S \sxto f T$ between finite sets extends to a linear map in a natural way, by defining $f \ket s := \ket {f(s)}$. We denote this graphically in the following simple way:
\begin{equation}
\begin{aligned}
\begin{tikzpicture}[thick, scale=\licsscale]
\draw (0,0) node [below] {$S$} to node [smallbox] {$f$} (0,2) node [above] {$T$};
\end{tikzpicture}
\end{aligned}
\end{equation}
The linear maps arising in this way can be characterized in an elegant way using the graphical notation, as morphisms satisfying the following two conditions:
\begin{calign}
\begin{aligned}
\begin{tikzpicture}[thick, yscale=-1, scale=\licsscale]
\draw (0,0) node [above] {$T$} to [out=up, in=\swangle] (0.7,1) node [smallcircle] {} to (0.7,1.75) node [smallbox, anchor=north] {$f$} to (0.7,3) node [below] {$S$};
\draw (1.4,0) node[above] {$T$} to [out=up, in=\seangle] (0.7,1);
\end{tikzpicture}
\end{aligned}
\quad=\quad
\begin{aligned}
\begin{tikzpicture}[thick, yscale=-1, scale=\licsscale]
\draw (0,0) node [above] {$T$} to (0,1.25) node [smallbox, anchor=south] {$f$} to [out=up, in=\swangle] (0.7,2.25) node [smallcircle] {} to (0.7,3) node [below] {$S$};
\draw (1.4,0) node [above] {$T$} to (1.4,1.25) node [smallbox, anchor=south] {$f$} to [out=up, in=\seangle] (0.7,2.25);
\end{tikzpicture}
\end{aligned}
&
\begin{aligned}
\begin{tikzpicture}[thick, scale=\licsscale]
\draw (0,0) node [below] {$S$} to (0,1.25) node [smallbox] {$f$} to (0,2.5) node [smallcircle] {};
\node [above] at (0,3) {$\phantom T$};
\end{tikzpicture}
\end{aligned}
\quad=\quad
\begin{aligned}
\begin{tikzpicture}[thick, scale=\licsscale]
\draw (0,0) node [below] {$S$} to (0,1.25) node [smallcircle] {};
\node [above] at (0,3) {$\phantom T$};
\end{tikzpicture}
\end{aligned}
\end{calign}
These are the \emph{comonoid homomorphism} conditions. The first says that if you  apply $f$ and then copy the result in the canonical basis of $T$, this is the same as copying the initial state in the canonical basis of $S$ and then applying $f$ to each branch. The second equation says that if you apply $f$ and then uniformly delete the result according to the canonical basis of $T$, the result is the same as simply deleting in the canonical basis of $S$.

A chosen element $x \in S$ corresponds to a function $1 \sxto x S$, which gives rise to a linear map with the following graphical representation:
\begin{equation}
\label{eq:xlinearmap}
\begin{aligned}
\begin{tikzpicture}[thick, scale=\licsscale]
\draw (0,0) node [smallbox] {$x$} to (0,1.0) node [above] {$S$};
\end{tikzpicture}
\end{aligned}
\end{equation}
The deletion operation $u^\dag$ introduced in~\eqref{eq:copydelete} acts in the same way on all elements, reducing them to the unit scalar:
\begin{equation}
\begin{aligned}
\begin{tikzpicture}[thick, scale=\licsscale]
\draw (0,0) node [smallbox] {$x$} to (0,1.0) node [smallcircle] {};
\end{tikzpicture}
\end{aligned}
\quad=\quad
1
\end{equation}
The linear map~\eqref{eq:xlinearmap} is related to its adjoint in the following way:
\begin{equation}
\begin{aligned}
\begin{tikzpicture}[thick, yscale=-1, xscale=-1, scale=\licsscale]
\draw (0,-1) node [above] {$S$} to (0,0) to [out=up, in=\swangle] (0.7,1);
\draw (1.4,0) to [out=up, in=\seangle] (0.7,1);
\draw (0.7,1) to (0.7,1.75) node [smallcircle] {};
\node [smallcircle] at (0.7,1) {};
\node [smallbox, anchor=south] at (1.4,0) {$x ^\dag$};
\end{tikzpicture}
\end{aligned}
\quad=\quad
\begin{aligned}
\begin{tikzpicture}[thick, scale=\licsscale]
\draw (0,0.25) node [smallbox] {$x$} to (0,1.5) node [above] {$S$};
\node [smallcircle, white] at (0,-1.3) {};
\end{tikzpicture}
\end{aligned}
\quad=\quad
\begin{aligned}
\begin{tikzpicture}[thick, yscale=-1, scale=\licsscale]
\draw (0,-1) node [above] {$S$} to (0,0) to [out=up, in=\swangle] (0.7,1);
\draw (1.4,0) to [out=up, in=\seangle] (0.7,1);
\draw (0.7,1) to (0.7,1.75) node [smallcircle] {};
\node [smallcircle] at (0.7,1) {};
\node [smallbox, anchor=south] at (1.4,0) {$x ^\dag$};
\end{tikzpicture}
\end{aligned}
\end{equation}
These equations say that the adjoint of~\eqref{eq:xlinearmap} is its \emph{dual}.

We can also use the topological calculus to express some simple features of particular linear maps. If a linear map $S \sxto f T$ arises from a function between underlying sets in the manner described above, then this function is \emph{injective} iff $f$ is an isometry, expressed by the following equation:
\begin{equation}
\begin{aligned}
\begin{tikzpicture}[thick, scale=\licsscale]
\draw (0,0) node [below] {$S$} to (0,1) node [smallbox] {$f$} to (0,2.25) node [smallbox] {$f ^\dagger$} to (0,3.25) node [above] {$S$};
\end{tikzpicture}
\end{aligned}
\,\,\quad=\quad
\begin{aligned}
\begin{tikzpicture}[thick, scale=\licsscale]
\draw (0,0) node [below] {$S$} to (0,3.25) node [above] {$S$};
\end{tikzpicture}
\end{aligned}
\end{equation}
Also, suppose a linear map $S \sxto g T$ arises from a function of underlying sets, and is \emph{evenly surjective}, meaning the preimage of every basis element of $T$ has the same cardinality $n$. Then the following equation holds:
\begin{equation}
\begin{aligned}
\begin{tikzpicture}[thick, scale=\licsscale]
\draw (0,0) node [smallcircle] {} to node [auto, swap, pos=0.35, inner sep=2pt] {$S$} (0,1) node [smallbox] {$g$} to (0,2) node [above] {$T$};
\end{tikzpicture}
\end{aligned}
\quad=\hspace{2pt}
\begin{aligned}
\begin{tikzpicture}[thick, scale=\licsscale]
\draw (0,1) node (dot) [smallcircle] {} to (0,2) node [above] {$T$};
\draw (0,0) node [smallcircle, white] {};
\node [anchor=east] at (dot.west) {$n$};
\end{tikzpicture}
\end{aligned}
\end{equation}
A lot more can be said on this topic, but these results are all we will need.

\subsection{Traces}
\label{sec:trace}

The topological algebra associated with the free Hilbert space $S$ on a finite set also allows us to describe taking traces of arbitrary linear maps $S \sxto L S$. We define the trace graphically in the following way:
\begin{equation}
\Tr(L) \quad=\quad
\begin{aligned}
\begin{tikzpicture}[thick, scale=\licsscale]
\draw (0,0.5) node [smallcircle] {} to (0,1) node [smallcircle] {} to [out=\nwangle, in=down] (-0.7,2) to (-0.7,2.5) to [out=up, in=\swangle] (0,3.5) node [smallcircle] {} to (0,4) node [smallcircle] {};
\draw (0,1) to [out=\neangle, in=down] (0.7,2) to (0.7,2.5) to [out=up, in=\seangle] (0,3.5);
\node (L) [smallbox] at (-0.7,2.25) {$L$};
\end{tikzpicture}
\end{aligned}
\end{equation}
The trace of the identity gives the size of the set $S$:
\begin{equation}
\begin{aligned}
\begin{tikzpicture}[thick, scale=\licsscale]
\draw (0,0.5) node [smallcircle] {} to (0,1) node [smallcircle] {} to [out=\nwangle, in=down] (-0.7,2) to [out=up, in=\swangle] (0,3.0) node [smallcircle] {} to (0,3.5) node [smallcircle] {};
\draw (0,1) to [out=\neangle, in=down] (0.7,2) to [out=up, in=\seangle] (0,3.0);
\end{tikzpicture}
\end{aligned}
\quad=\quad
|S|
\end{equation}
The formalism also allows a purely topological proof of the cyclic property:
\begin{equation}
\nonumber
\begin{aligned}
\begin{tikzpicture}[thick, scale=\licsscale, yscale=0.8]
\draw (0,0.5) node [smallcircle] {} to (0,1) node [smallcircle] {} to [out=\nwangle, in=down] (-0.7,2) to (-0.7,3.5) to [out=up, in=\swangle] (0,4.5) node [smallcircle] {} to (0,5.0) node [smallcircle] {};
\draw (0,1) to [out=\neangle, in=down] (0.7,2) to (0.7,3.5) to [out=up, in=\seangle] (0,4.5);
\node (L) [smallbox] at (-0.7,2.25) {$L$};
\node (L) [smallbox] at (-0.7,3.25) {$M$};
\end{tikzpicture}
\end{aligned}
\quad=\quad
\begin{aligned}
\begin{tikzpicture}[thick, scale=\licsscale, yscale=0.8]
\node (L) at (0,0) [smallbox] {$L$};
\node (M) at (2,0) [smallbox] {$M$};
\draw (L.north) to [out=up, in=\swangle] (0.5,1) node (d1) [smallcircle] {} to [out=\seangle, in=\nwangle] (1.5,-1) node [smallcircle] (d2) {} to [out=\neangle, in=south] (M.south);
\draw (M.north) to [out=up, in=\swangle] (2.5,1) node [smallcircle] (d3) {} to [out=\seangle, in=north] (3,0) to [out=down, in=\neangle] (1.5,-2.5) node [smallcircle] (d4) {} to [out=\nwangle, in=down] (L.south);
\draw (d1.center) to +(0, 0.5) node [smallcircle] {};
\draw (d3.center) to +(0, 0.5) node [smallcircle] {};
\draw (d2.center) to +(0, -0.5) node [smallcircle] {};
\draw (d4.center) to +(0, -0.5) node [smallcircle] {};
\end{tikzpicture}
\end{aligned}
\end{equation}
\begin{equation}
\quad=\quad
\begin{aligned}
\begin{tikzpicture}[thick, scale=\licsscale, yscale=0.8]
\node (L) at (0,0) [smallbox] {$L$};
\node (M) at (1.5,-0.5) [smallbox] {$M$};
\node [smallcircle] (d2) at (1.8,-1.5) {};
\node (d3) [smallcircle] at (1.2,0.75) {};
\draw (L.north) to [out=up, in=\swangle] (1.2,2) node (d1) [smallcircle] {} to [out=\seangle, in=up] (2.1,0) to (2.1,-0.8) to [out=down, in=up] ([xshift=-0.3cm, yshift=0.3cm] d2.center) to [out=down, in=\nwangle, in looseness=1]  (d2.center) to [out=\neangle, in=down, out looseness=1] ([xshift=0.3cm, yshift=0.3cm] d2.center) to [out=up, in=south] (M.south);
\draw (M.north) to [out=up, in=\swangle, in looseness=2] (d3.center) {} to [out=\seangle, in=north] ([xshift=0.3cm, yshift=-0.3cm] d3.center) to [out=down, in=up] (0.9,-0.25) to (0.9,-0.5) to [out=down, in=\neangle] (0.45,-1.5) node [smallcircle] (d4) {} to [out=\nwangle, in=down] (L.south);
\draw (d1.center) to +(0, 0.5) node [smallcircle] {};
\draw (d3.center) to +(0, 0.5) node [smallcircle] {};
\draw (d2.center) to +(0, -0.5) node [smallcircle] {};
\draw (d4.center) to +(0, -0.5) node [smallcircle] {};
\end{tikzpicture}
\end{aligned}
\quad=\quad
\begin{aligned}
\begin{tikzpicture}[thick, scale=\licsscale, yscale=0.8]
\draw (0,0.5) node [smallcircle] {} to (0,1) node [smallcircle] {} to [out=\nwangle, in=down] (-0.7,2) to (-0.7,3.5) to [out=up, in=\swangle] (0,4.5) node [smallcircle] {} to (0,5.0) node [smallcircle] {};
\draw (0,1) to [out=\neangle, in=down] (0.7,2) to (0.7,3.5) to [out=up, in=\seangle] (0,4.5);
\node (L) [smallbox] at (-0.7,2.25) {$M$};
\node (L) [smallbox] at (-0.7,3.25) {$L$};
\end{tikzpicture}
\end{aligned}
\end{equation}

\subsection{Matrix algebras}

\newcommand\name[1]{\ensuremath{\text{\textopencorner} #1 \text{\textcorner}}}
\noindent
The space of operators on a Hilbert space $\C ^n$ is canonically isomorphic to $\C ^n \otimes \C^n$. Given $\C^n \sxto L \C ^n$, we define its \emph{name} $\name{L} \in \C^n \otimes \C^n$ in the following way:
\begin{equation}
\begin{aligned}
\begin{tikzpicture}[thick, scale=\licsscale]
\node (L) [draw, minimum width=1.3cm] at (0,0) {\name L};
\draw (L.35) to ([yshift=0.5cm] L.35) node [above] {$\C ^n$};
\draw (L.145) to ([yshift=0.5cm] L.145) node [above] {$\C ^n$};
\node [smallcircle, white] at (0,-1.5) {};
\end{tikzpicture}
\end{aligned}
\quad:=\quad\hspace{-5pt}
\begin{aligned}
\begin{tikzpicture}[thick, yscale=-1, , scale=\licsscale, scale=0.7]
\draw (0,-1.5) node [above] {$\C ^n$} to (0,0) to [out=up, in=\swangle] (0.7,1);
\draw (1.4,0) to [out=up, in=\seangle] (0.7,1);
\draw (0.7,1) to (0.7,1.75) node [smallcircle] {};
\node [smallcircle] at (0.7,1) {};
\node (L) [smallbox, anchor=south] at (1.4,0) {$L$};
\draw (L.north) to (1.4,-1.5) node [above] {$\C ^n$};
\end{tikzpicture}
\end{aligned}
\end{equation}
The correspondence between operators and their names is known in quantum information as the Choi-Jamiolkowski isomorphism, and plays an important role in logic gate teleportation. There is an algebra operation
\begin{equation}
\text{comp} : (\C^n \otimes \C^n) \otimes (\C^n \otimes \C^n) \to \C^n \otimes \C^n
\end{equation}
on $\C^n \otimes \C^n$,  which is the algebra $\Mat(n)$ of $n$-by-$n$ matrices. It is given by the following diagram:
\begin{equation}
\begin{aligned}
\begin{tikzpicture}[thick, scale=0.7, scale=\licsscale]
\draw (0,0) to [out=up, in=\swangle] (0.7,1) node [smallcircle] {} to [out=\seangle, in=up] (1.4,0);
\draw (0.7,1) to (0.7,1.5) node [smallcircle] {};
\draw (-1.1,0) to [out=up, in=down] (0.15,2.5);
\draw (2.5,0) to [out=up, in=down] (1.25,2.5);
\end{tikzpicture}
\end{aligned}
\end{equation}
This has the property that $\text{comp}(\name g, \name f) = \name{g  \circ f}$, as the following graphical argument shows:
\begin{equation}
\begin{aligned}
\begin{tikzpicture}[thick, xscale=1.2, scale=\licsscale]
\draw (0,0.5) node [above] {$\C^n$} to [out=down, in=up] (-0.7,-1) to (-0.7,-2) to [out=down, in=\nwangle] (-0.35,-2.5) node [smallcircle] (A) {} to [out=\neangle, in=down] (0,-2) to node [smallbox] {$g$} (0,-1) to [out=up, in=\swangle] (0.35,-0.5) node [smallcircle] (B) {} to [out=\seangle, in=up] (0.7,-1) to (0.7,-2) to [out=down, in=\nwangle] (1.05, -2.5) node [smallcircle] (C) {} to [out=\neangle, in=down] (1.4,-2) to node [smallbox] {$f$} (1.4,-1) to [out=up, in=down] (0.7,0.5) node [above] {$\C ^n$};
\draw (A.center) to ([yshift=-0.35cm] A.center) node [smallcircle] {};
\draw (B.center) to ([yshift=0.35cm] B.center) node [smallcircle] {};
\draw (C.center) to ([yshift=-0.35cm] C.center) node [smallcircle] {};
\node [smallcircle, white] at (0,-3.85) {};
\end{tikzpicture}
\end{aligned}
\quad=\quad
\begin{aligned}
\begin{tikzpicture}[thick, scale=\licsscale]
\draw (-0.7,0) node [above] {$\C^n$} to [out=down, in=up] (0,-1) to node [smallbox] {$f$} (0,-2) to [out=down, in=\neangle] (-0.35,-2.5) node (A) [smallcircle] {} to [out=\nwangle, in=down] (-0.7,-2) to [out=up, in=\seangle] (-1.05,-1.5) node (B) [smallcircle] {} to [out=\swangle, in=up] (-1.4,-2) to (-1.4,-2.5) to node [smallbox] {$g$} (-1.4,-3.5) to [out=down, in=\neangle] (-1.75,-4) node [smallcircle] (C) {} to [out=\nwangle, in=down] (-2.1,-3.5) to (-2.1,-1.5) to [out=up, in=down] (-1.4,0) node [above] {$\C^n$};
\draw (A.center) to ([yshift=-0.35cm] A.center) node [smallcircle] {};
\draw (B.center) to ([yshift=0.35cm] B.center) node [smallcircle] {};
\draw (C.center) to ([yshift=-0.35cm] C.center) node [smallcircle] {};
\end{tikzpicture}
\end{aligned}
\quad=\quad
\begin{aligned}
\begin{tikzpicture}[thick, yscale=-1, scale=0.7, scale=\licsscale]
\draw (0,-2.75) node [above] {$\C ^n$} to (0,0) to [out=up, in=\swangle] (0.7,1);
\draw (1.4,0) to [out=up, in=\seangle] (0.7,1);
\draw (0.7,1) to (0.7,1.75) node [smallcircle] {};
\node [smallcircle] at (0.7,1) {};
\node (L) [smallbox, anchor=south] at (1.4,0) {$g$};
\draw (L.north) to (1.4,-2.75) node [above] {$\C ^n$};
\node [smallbox, anchor=south] at (1.4,-1.25) {$f$};
\node [smallcircle, white] at (0,3.45) {};
\end{tikzpicture}
\end{aligned}
\end{equation}
For convenience, we will use the following graphical simplifications:
\begin{calign}
\begin{aligned}
\begin{tikzpicture}[thick, scale=0.7, scale=\licsscale]
\draw (0,0) to [out=up, in=up, looseness=2] (1.4,0);
\node [smallcircle, white] at (0.7,1.5) {};
\end{tikzpicture}
\end{aligned}
\quad:=\quad
\begin{aligned}
\begin{tikzpicture}[thick, scale=0.7, scale=\licsscale]
\draw (0,0) to [out=up, in=\swangle] (0.7,1) node [smallcircle] {} to [out=\seangle, in=up] (1.4,0);
\draw (0.7,1) to (0.7,1.5) node [smallcircle] {};
\end{tikzpicture}
\end{aligned}
&
\begin{aligned}
\begin{tikzpicture}[thick, scale=0.7, yscale=-1, scale=\licsscale]
\draw (0,0) to [out=up, in=up, looseness=2] (1.4,0);
\node [smallcircle, white] at (0.7,1.5) {};
\end{tikzpicture}
\end{aligned}
\quad:=\quad
\begin{aligned}
\begin{tikzpicture}[thick, scale=0.7, yscale=-1, scale=\licsscale]
\draw (0,0) to [out=up, in=\swangle] (0.7,1) node [smallcircle] {} to [out=\seangle, in=up] (1.4,0);
\draw (0.7,1) to (0.7,1.5) node [smallcircle] {};
\end{tikzpicture}
\end{aligned}
\end{calign}
We will write $\Mat(n)$ and $\C^n \otimes \C^n$ interchangeably to refer to the matrix algebra we have described here.

\subsection{Group representations}

A group $G$ has a multiplication function $G \times G \sxto {\scriptstyle m} G$ and unit element $e \in G$, which we think of as a function $1 \sxto e G$ from the 1-point set. Linearizing these, we denote them graphically in the following way:
\begin{calign}
\begin{aligned}
\begin{tikzpicture}[thick, scale=\licsscale]
\draw (-0.7,-1) node [below] {$G$} to [out=up, in=\swangle] (0,0);
\draw (0.7,-1) node [below] {$G$} to [out=up, in=\seangle] (0,0);
\draw (0,0) to (0,1) node [above] {$G$};
\node (m) at (0,0) [circlelabel] {$m$};
\end{tikzpicture}
\end{aligned}
&
\begin{aligned}
\begin{tikzpicture}[thick, scale=\licsscale]
\draw [white] (-0.7,-1) to [out=up, in=\swangle] (0,0);
\draw [white] (0.7,-1) to [out=up, in=\seangle] (0,0);
\draw (0,0) to (0,1) node [above] {$G$};
\node (m) at (0,0) [circlelabel] {$e$};
\node [below, white] at (0,-1) {$G$};
\end{tikzpicture}
\end{aligned}
\end{calign}
In the case that $G$ is finite, its representations can be characterized as the maps $G \sxto \rho \Mat(n)$ with the following topological properties:
\begin{calign}
\label{eq:rhocopied}
\begin{aligned}
\begin{tikzpicture}[thick, scale=\licsscale]
\draw (-0.7,-1) node [below] {$G$} to [out=up, in=\swangle] (0,0);
\draw (0.7,-1) node [below] {$G$} to [out=up, in=\seangle] (0,0);
\draw (0,0) to (0,1);
\node (m) at (0,0) [draw, circle, inner sep=1pt, font=\scriptsize, fill=white] {$m$};
\node (rho) at (0,1) [smallbox] {$\rho$};
\draw (rho.60) to ([yshift=0.5cm] rho.60);
\draw (rho.120) to ([yshift=0.5cm] rho.120);
\node at ([yshift=0.5cm] rho.90) [anchor=south] {$\Mat(n)$};
\end{tikzpicture}
\end{aligned}
\hspace{-2pt}=\hspace{5pt}
\begin{aligned}
\begin{tikzpicture}[thick, scale=\licsscale]
\draw (0,-0.05) node [below] {$G$} to (0,1);
\draw (1.5,0.00) node [below] {$G$} to (1.5,1);
\draw [white] (0.75,0) to (0.75,2.25);
\node (r1) at (0,1) [smallbox] {$\rho$};
\node (r2) at (1.5,1) [smallbox] {$\rho$};
\draw (r1.60) to [out=up, in=up] (r2.120);
\draw (r1.120) to [out=up, in=down] (0.55,2.5) to (0.55,2.8);
\draw (r2.60) to [out=up, in=down] (0.95,2.5) to (0.95,2.8);
\node [above] at (0.75,2.8) {$\Mat(n)$};
\end{tikzpicture}
\end{aligned}
&
\begin{aligned}
\begin{tikzpicture}[thick, scale=\licsscale]
\draw (0,0) node [circlelabel] {$e$} to (0,1) node (r1) [smallbox] {$\rho$};
\draw (r1.60) to ([yshift=0.7cm] r1.60);
\draw (r1.120) to ([yshift=0.7cm] r1.120);
\node at (0.0,2) [above] {$\Mat(n)$};
\node [below, white] at (0,-0.8) {$G$};
\end{tikzpicture}
\end{aligned}
\hspace{-4pt}=\hspace{-6pt}
\begin{aligned}
\begin{tikzpicture}[thick, scale=\licsscale]
\draw (0,0) to (0,-1) to [out=down, in=down, looseness=2] (0.5,-1) to (0.5,0);
\node at (0.25,0) [above] {$\Mat(n)$};
\node [below, white] at (0.25,-2.8) {$G$};
\end{tikzpicture}
\end{aligned}
\end{calign}
This says that $\rho$ is a homomorphism from the group algebra $G$ to the matrix algebra $\Mat(n)$, and is equivalent to the ordinary algebraic requirement for a representation that $\rho(g g') = \rho(g)\rho(g')$ for group elements $g,g' \in G$. We define $\dim(\rho):=n$, the dimension of the vector space on which the representation acts.

\subsection{Representations as projective measurements}
\label{sec:repprojmeas}

For a finite group $G$, an irreducible representation $G \sxto \rho \Mat(n)$ is proportional to a partial isometry, with scale factor $\sqrt{{n}/{|G|}}$.
\begin{theorem}
\label{thm:repnorm}
For a finite group $G$ and an irreducible representation $G \sxto \rho \Mat(n)$, the following holds:
\begin{equation}
\begin{aligned}
\begin{tikzpicture}[thick]
\node (rho) at (0,2.5) [smallbox] {$\rho$};
\node (pdag) at (0,1) [smallbox] {$\rho^\dag$};
\draw (pdag.-60 |- 0,0) to (pdag.-60);
\draw (pdag.-120 |- 0,0) to (pdag.-120);
\draw (pdag.north)
    to node [auto, swap, smallfont] {$G$} (rho.south);
\draw (rho.60) to (rho.60 |- 0,3.5);
\draw (rho.120) to (rho.120 |- 0,3.5);
\node [above] at (0,3.5) {$\Mat(n)$};
\node [below] at (0,0) {$\Mat(n)$};
\end{tikzpicture}
\end{aligned}
\quad=\quad
\frac{|G|}{n} 
\begin{aligned}
\begin{tikzpicture}[thick]
\draw (-0.2,0) to (-0.2,3.5);
\draw (0.2,0) to (0.2,3.5);
\node [below] at (0,0) {\makebox[0pt]{$\Mat(n)$}};
\node [above] at (0,3.5) {\makebox[0pt]{$\Mat(n)$}};
\end{tikzpicture}
\end{aligned}
\end{equation}
\end{theorem}
\begin{proof}
A simple graphical transcription of Theorem~\ref{thm:algebraicrhonormalization}.
\end{proof}

\noindent
The group algebra of a finite group is isomorphic to $\bigoplus _\rho \Mat(\dim(\rho))$, where the sum is taken over equivalence classes of irreducible representations. As a result, the family of maps
\begin{equation}
\begin{aligned}
\begin{tikzpicture}[thick]
\draw node [below] at (0,0) {$G$} to (0,1);
\node [above] at (0,2) {$\Mat(n)$};
\node (rho) at (0,1) [smallbox] {$\rho$};
\node [anchor=east] at (rho.west) {$\displaystyle \sqrt{\frac{n}{|G|}}$};
\draw (rho.60) to (rho.60 |- 0,2);
\draw (rho.120) to (rho.120 |- 0,2);
\end{tikzpicture}
\end{aligned}
\end{equation}
define a projective measurement on the group algebra. This will be degenerate in general, since for a nonabelian group not all irreducible representations are 1-dimensional.

\subsection{Decomposing group algebras}

Group algebras are semisimple, which in the finite-dimensional case means that they are products of matrix algebras. By standard results from group representation theory, these matrix algebras are the matrix algebras on the irreducible representation spaces, where each equivalence class of irreducible representation is represented once. The multiplication operation for the group algebra can be decomposed in the following way, where the sum is over each equivalence class of irreducible representations, and where $|G|$ is the order of the group:
\begin{equation}
\label{eq:mdecomp}
\begin{aligned}
\begin{tikzpicture}[thick]
\draw (-0.7,-1.5) node [below] {$G$} to (-0.7,-1) to [out=up, in=\swangle] (0,0);
\draw (0.7,-1.5) node [below] {$G$} to (0.7,-1) to [out=up, in=\seangle] (0,0);
\draw (0,0) to (0,1) node [above] {$G$};
\node (m) at (0,0) [draw, circle, inner sep=1pt, font=\scriptsize, fill=white] {$m$};
\end{tikzpicture}
\end{aligned}
\quad=\quad
\frac{1}{|G|} \hspace{5pt} \sum_\rho \dim (\rho)
\begin{aligned}
\begin{tikzpicture}[thick]
\draw (0,-1.5) node [below] {$G$} to (0,-0.9) node (r1) [smallbox] {$\rho$};
\draw (1,-1.5) node [below] {$G$} to (1,-0.9) node (r2) [smallbox] {$\rho$};
\draw (0.5,1) node [above] {$G$} to (0.5, 0.4) node (r3) [smallbox] {$\rho ^\dag$};
\draw (r1.120) to [out=up, in=down, in looseness=0.7] (r3.-120);
\draw (r2.60) to [out=up, in=down, in looseness=0.7] (r3.-60);
\draw (r1.60) to [out=up, in=up, looseness=1.2] (r2.120);
\end{tikzpicture}
\end{aligned}
\end{equation}
The general form of the right-hand side here can be deduced from the first property of~\eqref{eq:rhocopied}, which says that each $\rho$ is an algebra homomorphism. The coefficients can then be obtained by applying Theorem~\ref{thm:repnorm}.

We can use equation~\eqref{eq:mdecomp} to show that the multiplication vertex also copies irreducible representations when attached to either of the lower legs. For example, we obtain the following identity when composing with the adjoint of an irreducible representation  $\sigma ^\dag$ onto the lower-right leg, for which we make use of Theorem~\ref{thm:repnorm}:
\begin{equation}
\label{eq:copyonleg}
\begin{aligned}
\begin{tikzpicture}[thick, xscale=-1, scale=0.8]
\draw (-0.7,-1.7) to (-0.7,-1) to [out=up, in=\swangle] (0,0);
\draw (0.7,-2.7) node [below] {$G \vphantom)$} to (0.7,-1) to [out=up, in=\seangle] (0,0);
\draw (0,0) to (0,1) node [above] {$G$};
\node (m) at (0,0) [draw, circle, inner sep=1pt, font=\scriptsize, fill=white] {$m$};
\node (s) at (-0.7,-1.7) [smallbox, anchor=north] {$\sigma ^\dag$};
\draw (s.-120) to (s.-120 |- 0,-2.7);
\draw (s.-60) to (s.-60 |- 0,-2.7);
\node at (-0.7,-2.7) [below] {\makebox[0pt]{$\Mat(n)$}};
\end{tikzpicture}
\end{aligned}
\hspace{5pt}=\hspace{5pt}
{\scriptstyle \frac{\sum_\rho \dim(\rho)}{|G|} }
\begin{aligned}
\begin{tikzpicture}[thick, xscale=-1, scale=0.8]
\draw (0,-1) to (0,-0.4) node (r1) [smallbox] {$\rho$};
\draw (1,-2) node [below] {$G \vphantom)$} to (1,-0.4) node (r2) [smallbox] {$\rho$};
\draw (0.5,1) to (0.5, 0.9) node (r3) [smallbox] {$\rho ^\dag$};
\node (s) at (-0.0,-1) [smallbox, anchor=north] {$\sigma ^\dag$};
\draw (r2.120) to [out=up, in=down, in looseness=0.7] (r3.-120);
\draw (r1.60) to [out=up, in=down, in looseness=0.7] (r3.-60);
\draw (r2.60) to [out=up, in=up, looseness=1.2] (r1.120);
\draw (r3.north) to (0.5,1.7) node [above] {$G$};
\draw (s.-120) to (s.-120 |- 0,-2);
\draw (s.-60) to (s.-60 |- 0,-2);
\node at (0,-2) [below] {\makebox[0pt]{$\Mat(n)$}};
\end{tikzpicture}
\end{aligned}
\hspace{5pt}=\hspace{5pt}
\begin{aligned}
\begin{tikzpicture}[thick, xscale=-1, scale=0.8]
\draw (1,-2) node [below] {$G \vphantom)$} to (1,-0.4) node (r2) [smallbox] {$\sigma$};
\draw (0.5,1) to (0.5, 0.9) node (r3) [smallbox] {$\sigma ^\dag$};
\node (s) at (-0.0,-1) [smallbox, anchor=north, white] {$\sigma ^\dag$};
\draw (r2.120) to [out=up, in=down, in looseness=0.7] (r3.-120);
\draw (r1.60) to [out=up, in=down, in looseness=0.7] (r3.-60);
\draw (r2.60) to [out=up, in=up, looseness=1.2] (r1.120);
\draw (r3.north) to (0.5,1.7) node [above] {$G$};
\draw (r1.120) to (s.-120 |- 0,-2);
\draw (r1.60) to (s.-60 |- 0,-2);
\node at (0,-2) [below] {$\Mat(n)$};
\end{tikzpicture}
\end{aligned}
\end{equation}
A similar expression holds for $\sigma ^\dag$ attached to the lower-left leg. Intuitively, the map $\sigma ^\dag$ is `pulled through' the vertex $m$ when acting on one of the lower legs, extending the observation in~\eqref{eq:rhocopied} that representations are copied when composed at the target of the multiplication.

\subsection{Normal subgroups}

\noindent
For a normal subgroup $H$ of a finite group $G$, we can construct a canonical projection
\begin{equation}
G \sxto q G/H
\end{equation}
from $G$ to the quotient group $G/H$. For an irreducible representation $G \sxto \rho \Mat(n)$, exactly one of the following two topological facts is true.
\begin{theorem}
\label{thm:graphicalnormalsubgroup}
Given a projection $G \sxto q G/H$ onto the quotient group of a normal subgroup $H \subseteq G$, then for an irreducible representation $G \sxto \rho \Mat(n)$, exactly one of the following is true:
\begin{enumerate}
\item The representation $\rho$ factors as
\allowdisplaybreaks[1]
\begin{align}
\begin{aligned}
\begin{tikzpicture}[thick, scale=\licsscale]
\node (tau) at (0,1.75) [smallbox] {$\rho$};
\draw (0,0) node [below] {$G$}
    to (tau.south);
\draw (tau.60) to (tau.60 |- 0,3.5);
\draw (tau.120) to (tau.120 |- 0,3.5);
\node [above] at (0,3.5) {\makebox[0pt]{$\Mat(n)$}};
\end{tikzpicture}
\end{aligned}
\hspace{5pt}\quad&=\quad
\begin{aligned}
\begin{tikzpicture}[thick, scale=\licsscale]
\node (rho) at (0,2.5) [smallbox] {$\tau$};
\node (pdag) at (0,1) [smallbox] {$q$};
\draw (0,0) node [below] {$G$}
    to (pdag.south);
\draw (pdag.north)
    to node [auto, swap] {$G/H$} (rho.south);
\draw (rho.60) to (rho.60 |- 0,3.5);
\draw (rho.120) to (rho.120 |- 0,3.5);
\node [above] at (0,3.5) {\makebox[0pt]{$\Mat(n)$}};
\end{tikzpicture}
\end{aligned}
\hspace{20pt}
\ignore{\begin{array}{c}
\text{\em for some representation}\\
\text{\em $G/H \sxto \tau \Mat(n)$ of $G/H$}
\end{array}
\hspace{-40pt}}
\intertext{for some irreducible representation $G/H \sxto \tau \Mat(n)$.
\item The representation $\rho$ does not factor via $G \sxto q G/H$ and}
\begin{aligned}
\begin{tikzpicture}[thick, scale=\licsscale]
\begin{pgfonlayer}{foreground}
    \node (dot) [smallcircle] at (0,1) {};
    \node (s) [smallbox, anchor=south, thick] at (0.7,2) {$q$};
    \node (pi) [smallbox, anchor=south, thick] at (-0.7,2) {$\rho$};
\end{pgfonlayer}
\draw (0,0.25)
        node (dot2) [smallcircle] {}
    to (0,1)
    to [out=\nwangle, in=south] (-0.7,2);
\node [above] at (-0.7,3.25) {\makebox[0pt]{$\Mat(n) \vphantom/$}};
\draw (0,1)
    to [out=\neangle, in=south] (s.south)
    to (0.7,3.25) node [above] {$G/H$};
\node at (0,0.65) [anchor=west] {$G$};
\draw (pi.60) to (pi.60 |- 0,3.25);
\draw (pi.120) to (pi.120 |- 0,3.25);
\end{tikzpicture}
\end{aligned}
\hspace{-10pt}
\quad&=\quad
0\,.
\end{align}
\end{enumerate}
\end{theorem}
\begin{proof}
Immediate from Theorems~\ref{thm:algebraicnormalsubgroup} and~\ref{thm:sumtheorem}.
\end{proof}

\section{Results from group theory}
\subsection{Introduction}

\noindent
In this appendix we collect proofs of some standard results in the theory of finite groups and their complex representations, which are of use as part of our main presentation.

\subsection{Normal subgroups}

\begin{theorem}
\label{thm:algebraicnormalsubgroup}
For a normal subgroup $H$ of a finite group $G$, and an irreducible representation $G \sxto \rho \GL(n)$, exactly one of the following properties holds:
\begin{enumerate}
\item The representation factors via the quotient group $G/H$ and
\[
\sum_{h \in H} \rho(h) = |H| \, \id_n;
\]
\item The representation does not factor via the quotient group {${G/H}$} and
\[
\sum_{h \in H} \rho(h) = 0. \hspace{29pt}
\]
\end{enumerate}
\end{theorem}

\noindent\begin{proof}
Suppose that $\rho$ restricts to a trivial representation on $H$. Then $\rho(g) = \rho(g)\rho(h) = \rho(gh)$ for all $g \in G$ and $h\in H$, so $\rho$ is well-defined on cosets, and factors via the quotient group $G/H$. So property 1 holds, and property 2 is clearly false.

Otherwise, suppose that $\rho$ does not restrict to a trivial representation on $H$. So both parts of property 1 must be false. By Clifford's theorem~\cite{cst09-cta} the restricted representation $\rho_H$ is of the form
\begin{equation}
\rho_H(h) = \bigoplus _{i} \sigma \big( g_i ^{\vphantom{-1}} h g_i ^{-1} \big)
\end{equation}
where $\sigma$ is an irreducible representation of $H$, and where $g_i$ is some finitely-indexed family of elements of $G$, such that the summands are pairwise nonisomorphic representations. None of these summands can be the trivial representation, since then they would all be trivial, which would violate our hypothesis. Thus $\rho_H$ does not contain the trivial representation as a subrepresentation. Property 2 of the theorem follows, since $\sum_{h\in H} \rho(h)$ is proportional to a projector onto a trivial subrepresentation.

\end{proof}

\begin{theorem}
\label{thm:sumtheorem}
For a projection $G \sxto q G/H ^l$ onto the left coset space of a subgroup $H \subseteq G$, then for a representation $G \sxto \rho \GL(n)$ we have
\begin{equation}
\displaystyle \sum _{h \in H} \rho (h) = 0
\quad\Rightarrow\quad \displaystyle \sum _{g \in G} \rho(g) \otimes q(g) = 0
\,.
\end{equation}
\end{theorem}
\begin{proof}
For a subgroup $H \subseteq G$, we can choose a representative $g_c \in G$ for each element of the coset space $c \in G / H ^l$, such that every element $g \in G$ is uniquely of the form $h g_{c}$ for some $c \in G/H$ and $h \in H \subseteq G$, and such that $q(g_c) = c$. As a result, the sum of property~2 above can be rewritten as
\begin{align}
\nonumber
&\hspace{20pt}\sum _{h \in H} \sum _{c \in G/H} {\rho(h g_c)} \otimes {q(g _c)}
\\
&= \Bigg( \bigg( \sum _{h \in H} \rho(h) \bigg) \otimes \id _{G/H} \Bigg)
\circ \Bigg( \sum _{c \in G/H} \rho(g_c) \otimes c \Bigg)
\end{align}
This makes the implication clear.
\end{proof}

\subsection{Normalization of irreducible representations}

\begin{theorem}
\label{thm:algebraicrhonormalization}
For a finite group $G$ and an irreducible representation $G \sxto \rho \Mat(n)$,
\begin{calign}
\rho \circ \rho ^\dag = \frac{|G|}{n} \, \id _{\Mat(n)} \, .
\end{calign}
\end{theorem}

\noindent\begin{proof}
Since irreducible representations are surjective as functions on the group algebra, it suffices to show
\[
\rho \circ \rho ^\dag \circ \rho = \frac{|G|} {n} \rho \, ,
\]
and since the group elements form a basis for the group algebra, it is enough to show that for all $g \in G$
\[
\rho \circ \rho ^\dag \circ \rho (g) = \frac {|G|} {n} \rho(g) \, .
\]
From the definition
\[
\langle g | \rho ^\dag(M) \rangle _G = \langle \rho(g) | M\rangle  _{\Mat(n)} = \Tr(\rho(g ^{-1}) \circ M)
\]
we see that
\[
\rho ^\dag (M) = \sum _{h \in G} \Tr (\rho(h ^{-1}) \circ M)\,\ket h\,,
\]
which for $M=\rho(g)$ gives
\begin{align}
\nonumber
\rho \circ \rho ^\dag \circ \rho(g) &= \sum _{h \in G} \Tr(\rho(h ^{-1}) \circ \rho(g) ) \rho(h)
\\
\nonumber
&= \sum _{h \in G} \Tr(\rho((gh) ^{-1}) \circ \rho( g)) \rho(gh)
\\
\label{eq:rhorhodag}
&= \rho(g) \sum _{h \in G} \Tr(\rho(h ^{-1})) \rho(h)\,.
\end{align}
Consider the element $X := \sum _{h \in G} \Tr(\rho(h ^{-1})) \rho(h)$ of $\Mat(n)$. It is an intertwiner for the irreducible representation $\rho$, since
\begin{align}
\nonumber
X \circ \rho(g) &= \sum _{h \in G} \Tr(\rho(h ^{-1})) \rho(h) \circ \rho(g)
\\
\nonumber
&= \sum _{h \in G} \Tr(\rho(h ^{-1})) \rho(hg)
\\
\nonumber
&= \sum _{h' \in G} \Tr(\rho(gh' {}^{-1} g ^{-1})) \rho(g h' g ^{-1}g)
\quad
\text{\scriptsize (for $h \mapsto g h' g^{-1}$)}
\\
\nonumber
&= \rho(g) \sum _{h' \in H} \Tr(\rho(h' {}^{-1}))\rho(h')
\\
&= \rho(g) \circ X
\end{align}
Hence by Schur's lemma we have $X = k \cdot \id _{\C ^n}$ for some $k \in \C$. By the normalization of character functions
\begin{align}
\Tr (X) &= \sum _{h \in G} \Tr(\rho(h ^{-1})) \Tr(\rho(h)) = |G|,
\end{align}
and hence $k = \frac{|G|}{n}$. Equation~\eqref{eq:rhorhodag} then gives $\rho \circ \rho^\dag = \frac{|G|}{n} \id _{\Mat(n)}$ as required.
\end{proof}

\end{document}